\documentclass[twocolumn]{IEEEtran}
\ifCLASSINFOpdf
  \usepackage[pdftex]{graphicx}
\else
  \usepackage[dvips]{graphicx}
\fi
\usepackage{amsmath}
\ifCLASSOPTIONcompsoc
\usepackage[caption=false,font=normalsize,labelfont=sf,textfont=sf]{subfig}
\else
  \usepackage[caption=false,font=footnotesize]{subfig}
\fi

\usepackage[square, numbers, comma, sort&compress]{natbib}
\newcommand{\overbar}[1]{\mkern 1.5mu\overline{\mkern-2mu#1\mkern-2mu}\mkern 1.5mu}
\usepackage{amsthm}
\newtheorem{theorem}{Theorem}
\newtheorem{lemma}{Lemma}
\usepackage{booktabs}
\usepackage{hyperref}
\usepackage{xcolor}
\usepackage{textcomp}
\usepackage{fancybox,framed}
\begin{document}
\title{Analysis of Signals via Non-Maximally Decimated Non-Uniform Filter Banks}

%\author{Michael~Shell,~\IEEEmembership{Member,~IEEE,}
%        John~Doe,~\IEEEmembership{Fellow,~OSA,}
%        and~Jane~Doe,~\IEEEmembership{Life~Fellow,~IEEE}% <-this % stops a space
%\thanks{M. Shell was with the Department
%of Electrical and Computer Engineering, Georgia Institute of Technology, Atlanta,

\author{Sandeep~Patel, Ravindra~Dhuli, and~Brejesh~Lall% <-this % stops a space
\thanks{Sandeep Patel and Brejesh Lall are with the Bharti School of Telecom Technology and Management, Indian Institute of Technology Delhi, New Delhi, India - 110 016
(e-mail: sanpatel90@gmail.com; brejesh@ee.iitd.ac.in).}% <-this % stops a space
\thanks{Ravindra Dhuli is with the Department of Electronics and Communications Engineering, Vellore Institute of Technology, Andhra Pradesh, India - 522 237 (email: ravindra.d@vitap.ac.in).}
\thanks{\textcopyright 2019 IEEE.  Personal use of this material is permitted.  Permission from IEEE must be obtained for all other uses, in any current or future media, including reprinting/republishing this material for advertising or promotional purposes, creating new collective works, for resale or redistribution to servers or lists, or reuse of any copyrighted component of this work in other works.}
\thanks{Digital Object Identifier 10.1109/TCSI.2019.2914302}}

\doublebox{\begin{minipage}{1\textwidth}
\begin{center}
\textbf{This is the accepted version of the paper}
\end{center}
\medskip
Link to IEEE version: \href{https://ieeexplore.ieee.org/abstract/document/8716585}{https://ieeexplore.ieee.org/abstract/document/8716585} 

\medskip

\textbf{Digital Object Identifier 10.1109/TCSI.2019.2914302}

\medskip

Citation: S. Patel, R. Dhuli and B. Lall, ``Analysis of Signals via Non-Maximally Decimated Non-Uniform Filter Banks," in IEEE Transactions on Circuits and Systems I: Regular Papers.

\medskip

\textcopyright 2019 IEEE.  Personal use of this material is permitted.  Permission from IEEE must be obtained for all other uses, in any current or future media, including reprinting/republishing this material for advertising or promotional purposes, creating new collective works, for resale or redistribution to servers or lists, or reuse of any copyrighted component of this work in other works.
\end{minipage}}

\clearpage

\maketitle

\begin{abstract}
This paper addresses the important problem of reconstructing a signal from multiple multirate observations. The observations are modeled as the output of an analysis bank, and time-domain analysis is carried out to design an optimal FIR synthesis bank. We pose this as a minimizing the mean-square problem and prove that at least one optimal solution is always possible. A parametric form for all optimal solutions is obtained for a non-maximally decimated filter bank. The necessary and sufficient conditions for an optimal solution, that results in perfect reconstruction (PR), are derived as time-domain pseudocirculant conditions. This represents a novel theoretical contribution in multirate filter bank theory. We explore PR in a more general setting. This results in the ability to design a synthesis bank with a particular delay in the reconstruction. Using this delay, one can achieve PR in cases where it might not have been possible otherwise. Further, we extend the design and analysis to non-uniform filter banks and carry out simulations to verify the derived results.
\end{abstract}

%\begin{IEEEkeywords}
%IEEE, IEEEtran, journal, \LaTeX, paper, template.
%\end{IEEEkeywords}

% \fi
\IEEEpeerreviewmaketitle

\section{Introduction}
%\IEEEPARstart{T}{his} demo file is intended to serve as a ``starter file''
%for IEEE journal papers produced under \LaTeX\ using
%IEEEtran.cls version 1.8b and later.
%% You must have at least 2 lines in the paragraph with the drop letter
%% (should never be an issue)
%I wish you the best of success.

%\subsection{Subsection Heading Here}
%Subsection text here.
%
%% needed in second column of first page if using \IEEEpubid
%%\IEEEpubidadjcol
%
%\subsubsection{Subsubsection Heading Here}
%Subsubsection text here.

\IEEEPARstart{M}{ultirate} filter banks have been used extensively for applications in the fields of image and video processing~\citep{jn:strelafb, jn:ChengDFB, jn:HsiangVideoFB}, communications~\citep{jn:cognitivefb,jn:CruzRoldn}, subband adaptive filtering~\citep{jn:shynk}, radar systems~\citep{jn:radar_xu}, etc. 
The filter bank theory evolved from the need to decompose a signal spectrum into multiple frequency subbands in order to efficiently perform a number of operations like compression, coding, noise cancellation, etc. Initially, two channel quadrature mirror filter bank (QMF) was proposed in the literature~\citep{conf:qmf_first_paper}, and later, it was extended to arbitrary $M$-channel case~\citep{jn:vaid_mqmf, jn:vaid_assp_mag}. To better exploit the signal characteristics in a number of applications like audio coding, spectral analysis, subband adaptive filtering, etc., non-uniform filter banks (NUFBs) having different channel bandwidth were introduced~\citep{jn:xie_nufb_adv,conf:hoang,conf:nayebi_nufb_timed}. In~\citep{jn:gadre}, filter banks with different multirate factors were used, to realize the filter transfer functions for reducing the multiplicative complexity. NUFBs with rational sampling factors were considered in~\citep{conf:hoang,jn:vetterli_nufb}. However, they were found to be equivalent 
to integer decimated NUFBs~\citep{conf:Akkarakaran}. A number of design methods
like direct method~\citep{jn:vetterli_nufb}, recombination method~\citep{jn:xie_nufb_adv}, modulation based method~\citep{jn:modulated_nufb}, etc. exist in the literature. Application specific design methods have also been proposed in the literature. For example, in~\citep{jn:acoustic}, an NUFB was designed for acoustic echo-cancellation by maximizing the signal-to-alias ratio. All these methods design both the analysis and synthesis banks of an NUFB. The analysis and synthesis filters in a filter bank can be IIR or FIR. FIR filters have advantage over IIR in respect of stability and linear-phase property. In the literature, various design methods using transform like Lapped Orthogonal Transform~\citep{jn:LoT_malvar}, Generalized Lapped Orthogonal Transform~\citep{jn:genlot}, Extended Lapped Transform~\citep{jn:malvar_fast_elt}, etc. have been proposed to design a perfect reconstruction (PR) FIR maximally decimated uniform filter bank (UFB). A comprehensive analysis of FIR perfect reconstruction filter banks can be found in~\citep{jn:vetterli_fir_pr}. It has been shown that oversampled filter banks (OFBs) offer additional design freedom as compared to critically sampled filter banks~\citep{jn:duval,jn:bolc_OFB_cosinepf}. The inbuilt redundancy present in them has been exploited for a variety of applications~\citep{jn:duhamel,jn:johansson2007flexible,jn:OFBAppLiu,jn:OFBRahimi}. A frame theoretic analysis of OFBs can be found in~\citep{jn:Cvetkovic}. Given an oversampled analysis bank, a parameterization for all synthesis banks resulting in PR was developed in~\citep{jn:bolcskeiframe}. Undersampled filter banks which provide underdetermined representation of signals have found application in compression of signals~\citep{conf:undersampled_compressive}. Despite the filter banks have been around for more than two decades, they continue to be applied to new and emerging areas. Graph signal processing is one such interesting area where filter banks techniques find applications in solving diverse problems~\citep{jn:dsp_graph}. 

The necessity in any scenario to estimate a signal or its properties from its subband components can be addressed well using the filter bank based processing. Expectedly, many research problems appear posing such requirement. Synthesis from subband components has been utilized in applications like subband coding~\citep{jn:vetterli_mfb,conf:rothweiler}, analog/digital conversion~\citep{jn:analog-digital-conversion}, transmultiplexer~\citep{jn:vetterli_mfb}, ultrawideband receiver~\citep{1198100}, TV white space transceiver~\citep{6334306}, channel coding~\citep{Weiss2006}, etc. To get better results, processing on subband components, rather than on full-band signal, have been carried out in problems like speech enhancement~\citep{jn:speech_enhancement_wu}, ECG processing~\citep{conf:afonso_ecg}, subband adaptive filtering~\citep{jn:lee_nlms}, seismic signals processing~\citep{conf:duval_seismic}, etc. The processed subband signals are then combined using a synthesis bank to get the final output. In multirate sensors networks, sensor outputs are modeled as an analysis bank output, and these outputs are fused to get a unified measurement~\citep{jn:jahromisensor}. In~\citep{conf:scrofani_icassp}, construction of a high-resolution signal from multiple low-resolution signals has been carried out. The problem of linear prediction from multirate observations has been examined in~\citep{conf:therrien_prediction}. In many applications, rather than a signal, a property of interest is to be estimated from multirate observations. Time delay estimation from low rate samples is an example of that, and it finds application in speech localization, processing using microphone arrays, radar, underwater acoustics, wireless communications and other areas~\citep{jn:tdoa_gedalyahu, conf:jahromi_tdoa}. Direction of arrival estimation is another example which finds application in wireless communications, medical imaging, seismology, etc.~\citep{5278497}. In [31], estimation of the spectrum of a signal from its low resolution signals was carried out. This is required when a signal can be observed only indirectly using low-rate sensors. 

To address the problem of estimating a signal from its subbands, various solutions using Wiener filtering, Kalman filtering, $\mathcal{H}^\infty$ filtering, compressive sensing, adaptive filtering, etc. have been proposed in the literature. Wiener filtering based methods require joint stationarity of the subband input and desired signals along with the knowledge of second order statistics of them. Kuchler and Therrien~\citep{conf:kuchler} used Wiener filtering to estimate a signal from multiple multirate observations. Scrofani and Therrien~\citep{conf:scrofani_icassp} extended the analysis to two-dimensional signals. Vaidyanathan and Chen~\citep{conf:chen} designed a statistically optimal synthesis bank for a subband coder by using the biorthogonal solution followed by a matrix Wiener filter in the synthesis stage. In case of Kalman filtering, the statistics of the input and noise signals are required to be known even if that vary with time. Multirate Kalman filter based methods can be found in~\citep{jn:chen_multirate_kalman, jn:chen_kalman_nufb, conf:kalman_2D}. $\mathcal{H}^\infty$ optimization based methods, on the other hand, do not require any statistical assumption and perform worst-case design~\citep{conf:vikalo_hinf}. As $\mathcal{H}^\infty$ solution is generally not unique, mixed $\mathcal{H}^2/\mathcal{H}^\infty$ optimization methods~\citep{conf:vikalo_mixedh2hinf,conf:bora_hinf} have also been proposed in the literature. The complexity of $\mathcal{H}^2$ and $\mathcal{H}^2/\mathcal{H}^\infty$ solutions are very high since they are solved as linear matrix inequalities (LMIs)~\citep{conf:bora_hinf}. A least square based solution for the problem was presented in~\citep{conf:hawes_leastsqu}. However, the solution presented  was for two channels only. In~\citep{conf:chai_ssim}, Li Chai et al. formulated the problem using structural similarity criterion and obtained closed-form solution under certain assumptions like the filter length is equal to the decimation factor, the mean of the source is zero, etc. Compresive sensing based method was used in~\citep{5605572} to estimate a signal from multirate observations. However, the signal has to be sparse. In~\citep{conf:hawes_adaptive,jn:delopoulos}, adaptive filters were used to estimate a signal from two observation sequences. An adaptive least square lattice filter was used in~\citep{conf:tanc_adaptive}. However, the filter obtained was periodically time-varying. Statistical analysis of the estimation from the subbands has been performed by a number of authors. For example, in~\citep{jn:jahromiesti3e, jn:Tanc201093, conf:amini_homotopy}, power spectral density of a signal was estimated from multiple observations. In~\citep{conf:JahromiInformation}, a measure for information contained in a multirate observations was defined.

\subsection{Contribution}
In this paper, the problem of synthesizing a signal from its multirate observations is addressed by performing time-domain analysis. This approach has led to some additional novel contributions to the field which are now briefly presented. Under time-domain analysis, the matrix representation for various operations like filtering, downsampling, upsampling, etc. are developed and used to parameterize synthesis bank solutions. This parameterization is consistent and at least one optimum synthesis filter bank is guaranteed. We design the synthesis bank for a more general set up in comparison to existing methods~\citep{conf:hawes_leastsqu,conf:chai_ssim,5605572,conf:hawes_adaptive,jn:delopoulos}. The synthesis filter length can be arbitrary but given. Further, we have used delay in the reconstruction as a parameter to have flexibility in the synthesis filter design, and using that, we obtain better reconstruction error than possible with a zero delay design. Another advantage of our analysis is that we obtain a linear time-invariant (LTI) optimal matrix filter unlike the linear periodically time-varying (LPTV) optimal filters obtained in earlier works~\citep{conf:kuchler,conf:tanc_adaptive}. Our method always results in a solution in comparison to LMI based methods~\citep{conf:vikalo_mixedh2hinf,conf:bora_hinf} where a solution for an LMI may not be feasible~\citep{VANANTWERP2000363}. The necessary and sufficient condition on a synthesis bank solution to result in PR is derived. Under this condition, the polyphase matrices for the analysis and synthesis banks satisfy the general pseudocirculant property~\citep{bk:vaidyn,jn:vaidypoly}, unlike the limited identity matrix condition in the earlier works~\citep{jn:duval,jn:Cvetkovic}. The necessary and sufficient conditions on the analysis bank and the delay for a PR synthesis bank solution to exist are derived. This analysis forms an alternative to the frame theory based approach in previous works~\citep{jn:Cvetkovic, jn:bolcskeiframe}. Time-domain pseudocirculant result derived in this paper form an important theoretical contribution to the multirate theory. The derivation of a range of delay values for which PR is possible is another contribution. We have also extended the design methodology for any arbitrary NUFB. Under experimental results, a number of experiments are presented to highlight contributions and to show application of our work. 

\subsection{Organization}
This paper is organized as follows: Section~\ref{sec_matrixrep} discusses the various matrix representations for time-domain signals and operations. In Section~\ref{sec_wiener}, we introduce the framework used to design the synthesis filter for a given analysis bank. In Section~\ref{sec_firpr}, we derive the pseudocirculant condition in time-domain and present the necessary and sufficient conditions on an analysis bank and on the delay to result in PR. We provide the experimental results in Section~\ref{sec_exp} and finally, concluding remarks are given in Section~\ref{sec_conclusion}. 

\subsection{Notation}
The $(i,j)$-th entry of a matrix $\mathbf{B}$ is given by $B_{i,j}$ or $B(i,j)$. The transpose and conjugate-transpose of a matrix or vector quantity are denoted by the superscripts $(.)^T$ and $(.)^H$, respectively. An identity matrix of size $P$ is represented by $\mathbf{I}$ when the size is implicit and by $\mathbf{I}_P$ otherwise. The pseudoinverse of a matrix $\mathbf{A}$ is denoted by $\mathbf{A}^\dagger $. The superscript $(.)^*$ stands for conjugation of a quantity. Convolution of two sequences is denoted by $\ast$. The ceiling function is represented by $\left\lceil.\right\rceil$.

\section{Matrix Representations}
\label{sec_matrixrep}
In this section, we develop the matrix representations for various multirate operations, which are required for an effective time-domain analysis of multirate filter banks.

\subsection{Time-domain Sequence Representation}
A vector signal is developed from the given signal $v(n)$ by stacking the past $P$ samples as:
\begin{equation}
\mathbf{v}(n) = \begin{bmatrix}
v(n-P+1) & v(n-P+2) & \dots & v(n)
\end{bmatrix}^T.
\end{equation}
It is termed as an observation vector for the signal $v(n)$~\citep{master:koup}. Its flipped version or reversal can be obtained as
\begin{equation}
\overbar{\mathbf{v}}(n) = \mathbf{J}_P \mathbf{v}(n),
\end{equation}
where $\mathbf{J}_P$ is a counter-identity or reverse-identity matrix. For a matrix $\mathbf{B}$, reversal operation results in reversing of order along both rows and columns. Mathematically, the same can be expressed as
\begin{equation}
\overbar{\mathbf{B}} = \mathbf{J}_M \mathbf{B} \mathbf{J}_N
\end{equation}
if the size of the matrix is $M \times N$. For two matrices $\mathbf{A}$ and $\mathbf{B}$ whose product is defined, we have the following relation:
\begin{equation}
\overbar{\mathbf{A} \mathbf{B}} = \overbar{\mathbf{A}} \overbar{\mathbf{B}}.
\end{equation} 

\subsection{Decimation}
Consider a system, as shown in Fig.~\ref{fig_dec_timead}, with a $M$-fold decimator preceded by a time-advance by $l$. The input and output of it are related as
\begin{equation}
y(n) = x(Mn+l).
\label{eq_decirel}
\end{equation}
\begin{figure}[htbp]
\centering
\includegraphics[scale = 1]{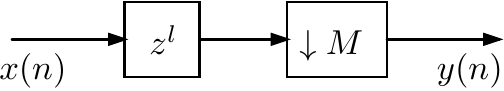}
\vspace{-8pt}
\caption{A decimator preceded by time-advance}
\label{fig_dec_timead}
\end{figure}
Such a system can be used to obtain polyphase components of a signal when $l$ is in the range $[0, M-1]$. Two matrix representations are possible for a system. If an observation vector 
\begin{equation}
\mathbf{y} = \begin{bmatrix}
y(0) & y(1) & \dots & y(P-1)
\end{bmatrix}^T
\end{equation}
for the output is to be evaluated, then the same can be done as follows:
\begin{equation}
\arraycolsep=2pt
\medmuskip=0mu
\begin{split}
\begin{bmatrix}
y(0) \\
y(1) \\
\vdots \\
y(P-1)
\end{bmatrix} &= \begin{bmatrix}
1 & 0 & \dots & 0 & \dots  & 0 \\
0 & 0 & \dots & 1 & \dots  & 0 \\
\vdots & \vdots & \ddots & \vdots & \ddots  & \vdots \\
0 & 0 & \dots & 0 & \dots & 1 
\end{bmatrix} \begin{bmatrix}
x(l) \\
x(l+1) \\
\vdots \\
x(M+l) \\
\vdots \\
x(M(P-1) + l)
\end{bmatrix}, \\
\mathbf{y} &= \mathbf{D}_{P , M(P-1)+1}^0 \mathbf{x} ,
\end{split}
\end{equation}
where $\mathbf{D}_{P , M(P-1)+1}^0$ is known as a decimation matrix. On the other hand, if an input sample vector 
\begin{equation}
\mathbf{x} = \begin{bmatrix}
x(0) & x(1) & \dots & x(P-1)
\end{bmatrix}^T
\end{equation}
is given, then the corresponding output samples can be obtained as
\begin{equation}
\mathbf{y} = \mathbf{D}_{\lceil \frac{P}{M}  \rceil , P}^l \mathbf{x},
\end{equation} 
where the elements of the decimation matrix $\mathbf{D}_{\lceil P/M  \rceil , P}^l$ are given by
\begin{equation}
\begin{split}
D_{\lceil \frac{P}{M}  \rceil , P}^l(i,j) &= \begin{cases}
1, \qquad &\text{if } j = Mi+l  \\
0, \qquad &\text{otherwise}.
\end{cases}
\end{split}
\end{equation}
Here $l$ is in the range $0 \leq l \leq M-1$. For $l \geq M$, the same matrix representations can be extended.

\subsection{Expansion}
If an input observation vector
\begin{equation}
\mathbf{x} = \begin{bmatrix}
x(0) & x(1) & \dots & x(P-1)
\end{bmatrix}^T 
\end{equation}
is available, then the output of a $M$-fold expander can be obtained as
\begin{equation}
\begin{split}
\mathbf{y} &= \mathbf{U}_{M(P-1)+1 , P} \mathbf{x},
\end{split}
\end{equation}
where $\mathbf{U}_{M(P-1)+1 , P}$ is an expansion matrix. This matrix is related to a decimation matrix as
\begin{equation}
\mathbf{U}_{M(P-1)+1 , P} = (\mathbf{D}_{P , M(P-1)+1}^0 )^T.
\end{equation}
It can be easily verified that
\begin{equation}
\begin{split}
\mathbf{D}_{P , M(P-1)+1}^0  \mathbf{U}_{M(P-1)+1 , P} &= \mathbf{I}_P, \\
\mathbf{U}_{M(P-1)+1 , P} \mathbf{D}_{P , M(P-1)+1}^0 &\neq \mathbf{I}_{M(P-1)+1}.
\end{split}
\end{equation}

\subsection{Convolution}
Consider a causal linear time-invariant (LTI) filter with impulse response $\{h(0),\allowbreak h(1),\allowbreak \dots, \allowbreak h(Q-1)\}$. Its input and output are related as
\begin{equation}
y(n) = \sum_{k=0}^{Q-1} h(k) x(n-k).
\end{equation}
We present the matrix representations for the above relation for two use cases. Under first case, consider an observation vector 
\begin{equation}
\mathbf{y}(n) = \begin{bmatrix}
y(n-P+1) & y(n-P+2) & \dots & y(n)
\end{bmatrix}^T
\end{equation}
for the output which needs to be evaluated. The same can be achieved using
\begin{equation}
\mathbf{y}(n) = \mathbf{H} \mathbf{x}(n),
\label{eq_conv_form1}
\end{equation}
where
\begin{equation}
\arraycolsep = 2pt
\medmuskip=0mu
\mathbf{x}(n) = \begin{bmatrix}
x(n-P+2-Q) & x(n-P+3-Q) & \dots & x(n)
\end{bmatrix}^T,
\end{equation}
and $\mathbf{H}$ is a convolution matrix given by
\begin{equation}
\medmuskip=0mu
\arraycolsep=2pt
\mathbf{H} = \begin{bmatrix}
h(Q-1) & h(Q-2) & \dots & h(0) & \dots  & 0 \\
0 & h(Q-1) & \dots & h(1) & \dots & 0 \\
\vdots & \vdots & \ddots & \vdots & \ddots & \vdots  \\
0 & 0 & \dots & 0 & \dots  & h(0)
\end{bmatrix}.
\end{equation}
The second representation is obtained for a scenario where an input observation vector is given, and the corresponding output samples need to be evaluated. Consider an input observation vector 
\begin{equation}
\mathbf{x} = \begin{bmatrix}
x(0) & x(1) & \dots & x(P-1)
\end{bmatrix}^T,
\end{equation}
then the complete output due to it can be expressed as
\begin{equation}
\thickmuskip = 2mu
\arraycolsep = 1pt
\medmuskip=0mu
\resizebox{1\linewidth}{!}{$
\begin{bmatrix}
y(0) \\
y(1) \\
\vdots \\
y(Q-1) \\
\vdots \\
y(P+Q-2)
\end{bmatrix} = \begin{bmatrix}
h(0) & 0 & \dots & 0 \\
h(1) & h(0) & \dots & 0 \\
\vdots & \vdots & \ddots & \vdots \\
h(Q-1) & h(Q-2) & \dots & 0 \\
\vdots & \vdots & \ddots & \vdots \\
0 & 0 & \dots & h(Q-1) \\
\end{bmatrix} \begin{bmatrix}
x(0) \\
x(1) \\
\vdots \\
x(P-1)
\end{bmatrix}$} 
\end{equation}
or 
\begin{equation}
\mathbf{y} = \tilde{\mathbf{H}} \mathbf{x}.
\label{eq_conv_form2}
\end{equation}
It is implicit in the above representation that $x(n) = 0$ for $n \notin [0, P-1]$. The utility of the representation is in evaluation  of convolution of two filters. The two convolution matrices, so far obtained, are related with each other as:
\begin{equation}
\tilde{\mathbf{H}} = \overbar{\mathbf{H}}^T,
\end{equation}
where $\overbar{\mathbf{H}}$ denotes the reversal of the matrix.

In the following section, we address the design of a synthesis filter bank using the time-domain representations discussed so far.

\begin{figure*}[!t]
\centering
\includegraphics[scale = 0.76]{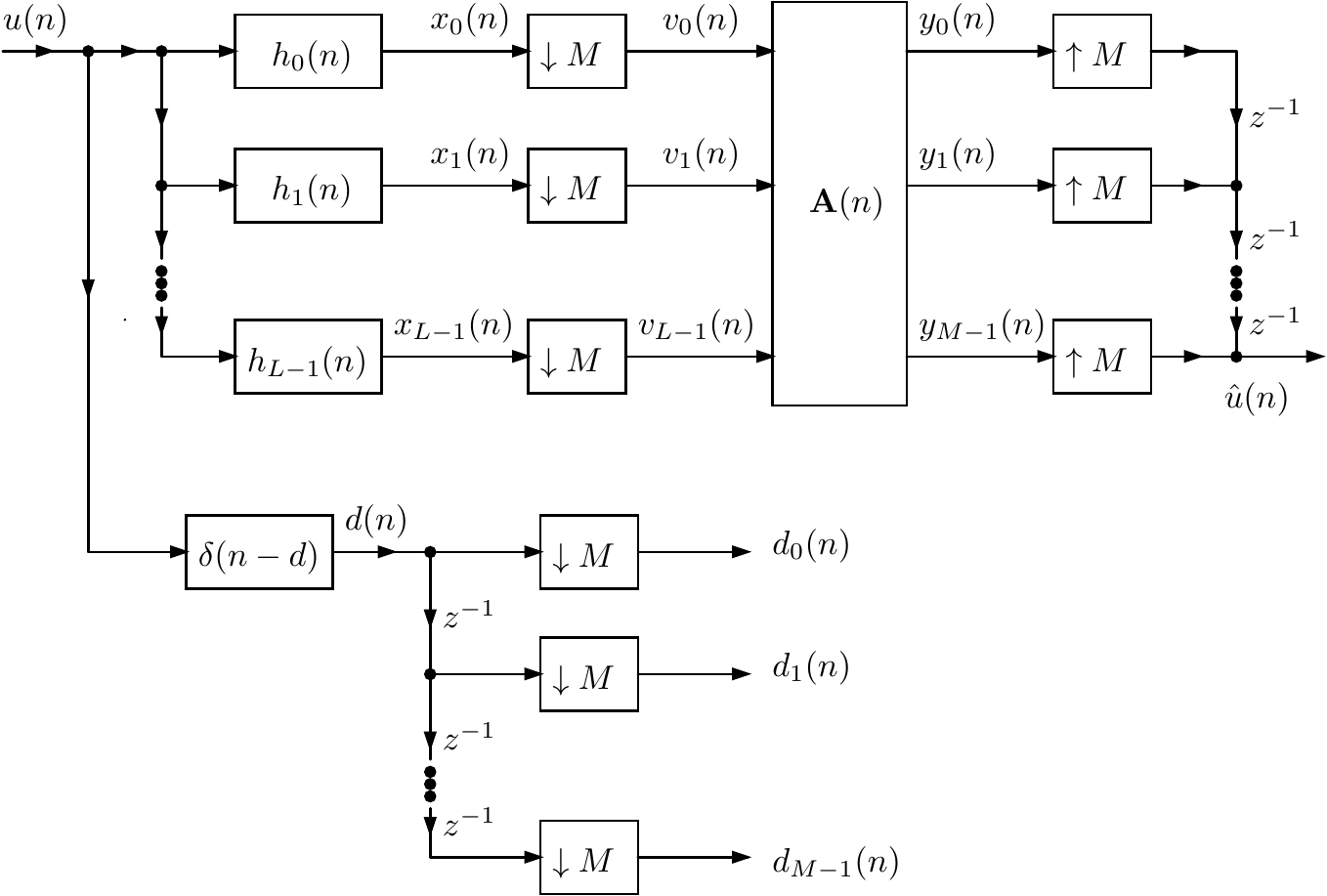}
\vspace{-10pt}
\caption{Matrix Wiener filter based synthesis stage for a UFB}
\label{fig_wfsynth}
\vspace{-12pt}
\end{figure*}
\section{Design Framework}
\label{sec_wiener}
We have multirate observations, derived from various sub bands, and we aim to reconstruct or estimate the original high resolution signal. Initially, we consider observations to be at the same low rate and model them as output of a uniform analysis bank. The case when the observations are at different rates is addressed later in this paper. The subbands are applied to a matrix synthesis filter, as shown in Fig.~\ref{fig_wfsynth}, where reconstruction happens at a low rate to reduce computational cost. The output of the synthesis filter is then unblocked~\citep{conf:rdhulicyclo} to get the final high resolution signal. We model our problem of designing the matrix synthesis filter bank in the framework of Wiener filtering. The Wiener filter is an optimum filter to minimize the mean-square error (MSE). For a multi-input multi-output system, the solution is known as a matrix Wiener filter. The details of the design setup are explained next.

Consider the filter bank shown in Fig.~\ref{fig_wfsynth}. A wide-sense stationary (WSS) input signal $u(n)$ passes through a given analysis bank consisting of the filters $\left\{h_i(n)\right\}_{i=0}^{L-1}$. The analysis filters are assumed to be FIR and causal. The output signals $\left\{x_i(n)\right\}_{i=0}^{L-1}$ of the analysis bank are uniformly decimated by the factor $M$ to obtain the subband signals $\left\{v_i(n)\right\}_{i=0}^{L-1}$ which are then applied to a matrix synthesis filter $\mathbf{A}(n)$. We choose the synthesis filter to be FIR and causal, and formulate a Wiener filtering problem for a given length of the synthesis filter. The desired signal for the problem is derived from the input signal $u(n)$ by adding a given delay $d$ and then performing blocking operation~\citep{conf:rdhulicyclo}. The delay is added to compensate the time lag that appears in the input signal when it passes through the filter bank. The blocking operation ensures that the desired signal $\mathbf{d}(n)$ is jointly stationary with the subband signal $\mathbf{v}(n)$~\citep{jn:sanelsevier}\textemdash a necessary condition required for application of Wiener filtering~\citep{bk:hayes}.

\subsection{FIR Matrix Synthesis Filter Expression}
In Appendix~\ref{app_matder}, we obtain the following expression for an FIR matrix Wiener filter with length $P$:
\begin{equation}
\mathbf{r}_{d_i v} = \mathbf{a}_i^T \mathbf{R}_{vv}, \qquad 0 \leq i \leq M-1,
\label{eq_apufb}
\end{equation}
where $\mathbf{R}_{vv}$ is the auto-correlation matrix of the input signal $\mathbf{v}(n)$ to the filter, $\mathbf{r}_{d_i v}$ is the cross-correlation vector between the desired signal component $d_i(n)$ and the signal $\mathbf{v}(n)$, and $\mathbf{a}_i$ is the $i$-th row of the filter. This expression can be simplified in the context of the UFB given in Fig.~\ref{fig_wfsynth}. In order to do that, a vector $\mathbf{v}_{\text{s}}(n)$ is created by stacking the observation vectors of the subband signals as
\begin{equation}
\mathbf{v}_{\text{s}}(n) = \begin{bmatrix}
\overbar{\mathbf{v}}_0^T(n) & \overbar{\mathbf{v}}_1^T(n) & \dots & \overbar{\mathbf{v}}_{L-1}^T(n)
\end{bmatrix}^T,
\end{equation}
where $\mathbf{v}_j(n)$ is an observation vector representing the $j$-th input to the Wiener filter at time $n$ and is given by
\begin{equation}
\arraycolsep = 2pt
\medmuskip=0mu
\begin{split}
\mathbf{v}_j(n) = &\begin{bmatrix}
v_j(n-P+1) &  v_j(n-P+2) &  \dots & v_j(n)
\end{bmatrix}^T.
\end{split}
\end{equation}
The quantities $\mathbf{R}_{vv}$ and $\mathbf{r}_{d_i v}$ can be expressed in terms of $\mathbf{v}_{\text{s}}(n)$ as
\begin{equation}
\begin{split}
\mathbf{R}_{vv} &= \mathbf{E}[\mathbf{v}_{\text{s}}(n) \mathbf{v}_{\text{s}}^H(n)], \\
\mathbf{r}_{d_i v} &= \mathbf{E}[ d_i(n) \mathbf{v}_{\text{s}}^H(n)].
\end{split}
\end{equation}
To make use of the Wiener filter in~(\ref{eq_apufb}), we derive an expression of the vector $\mathbf{v}_{\text{s}}(n)$ in terms of the input to the filter bank. We can note from Fig.~\ref{fig_wfsynth} that $v_j(n)=x_j(Mn)$ and this relationship can also be expressed in matrix form as
\begin{equation}
\begin{split}
\mathbf{v}_j(n) &= \mathbf{D}^0_{P , M(P-1)+1} \mathbf{x}_j(Mn),
\end{split}
\end{equation}
where
\begin{equation}
\arraycolsep = 2pt
\medmuskip=0mu 
\resizebox{1\linewidth}{!}{$
\mathbf{x}_j(Mn) = \begin{bmatrix}
x_j(M(n-P+1)) & x_j(M(n-P+1)+1) & \dots & x_j(Mn)
\end{bmatrix}^T.$}
\end{equation}
Further, the vector $\mathbf{x}_j(Mn)$ can be expressed in terms of the input $u(n)$. We start with an assumption that all the analysis filters are causal and have equal length $Q$. It can be assured by appropriate zero padding. Now we can obtain $\mathbf{x}_j(Mn)$ as
\begin{equation}
\mathbf{x}_j(Mn) = \mathbf{H}_j \mathbf{u}(Mn),
\end{equation}
where $\mathbf{H}_j$ is a convolution matrix. The entire input to the Wiener filter can be expressed as follows:
\begin{equation}
\thickmuskip = 2mu
\medmuskip=0mu
\begin{split}
\mathbf{v}_{\text{s}}(n) = \begin{bmatrix}
\overbar{\mathbf{v}}_0(n) \\
\overbar{\mathbf{v}}_1(n) \\
\vdots \\
\overbar{\mathbf{v}}_{L-1}(n)
\end{bmatrix} &= \begin{bmatrix}
\mathbf{D}^0_{P , M(P-1)+1} \overbar{\mathbf{H}}_0 \\
\mathbf{D}^0_{P , M(P-1)+1} \overbar{\mathbf{H}}_1 \\
\vdots \\
\mathbf{D}^0_{P , M(P-1)+1} \overbar{\mathbf{H}}_{L-1}
\end{bmatrix} \overbar{\mathbf{u}}(Mn), \\ 
&= \mathbf{K} \overbar{\mathbf{u}}(Mn),
\end{split}
\label{eq_Kdef}
\end{equation}
where $\mathbf{K}$ is a matrix of size $LP \times [M(P-1)+Q]$. Using this result, the auto-correlation matrix $\mathbf{R}_{vv}$ can be written as
\begin{equation}
\begin{split}
\mathbf{R}_{vv} &= \mathbf{K} \mathbf{E}[\overbar{\mathbf{u}}(Mn) \overbar{\mathbf{u}}^H(Mn)] \mathbf{K}^H, \\
&= \mathbf{K} \mathbf{R}_{uu} \mathbf{K}^H .
\end{split}
\end{equation}
Similarly the cross-correlation vector $\mathbf{r}_{d_i v}$ becomes
\begin{equation}
\begin{split}
\mathbf{r}_{d_i v} &= \mathbf{E}[u(Mn-i-d) \overbar{\mathbf{u}}^H(Mn)] \mathbf{K}^H, \\
&= \mathbf{r}_{uu}^{i+d} \mathbf{K}^H,
\label{eq_rdiv}
\end{split}
\end{equation}
where $\mathbf{r}_{uu}^{i+d}$ is a row vector given by
\begin{equation}
\arraycolsep=2pt
\medmuskip=0mu
\mathbf{r}_{uu}^{i+d} = \begin{bmatrix}
r_{uu}(-i-d) \\
r_{uu}(-i-d + 1) \\
\vdots \\
r_{uu}(-i - d +M(P - 1) + Q - 1)
\end{bmatrix}^T. 
\end{equation}
By substituting these results in~(\ref{eq_apufb}), we obtain
\begin{equation}
\mathbf{r}_{uu}^{i+d} \mathbf{K}^H = \mathbf{a}_i^T \mathbf{K} \mathbf{R}_{uu} \mathbf{K}^H, \qquad 0 \leq i \leq M-1.
\end{equation}
The resulting equation is of form
\begin{equation}
\mathbf{A} \mathbf{a}_i = \mathbf{b}_i,
\label{eq_aisysofeq}
\end{equation}
where
\begin{equation}
\begin{split}
\mathbf{A} &= \left( \mathbf{K} \mathbf{R}_{uu} \mathbf{K}^H \right) ^T, \\
\mathbf{b}_i &= \left( \mathbf{r}_{uu}^{i+d} \mathbf{K}^H \right) ^T.
\end{split}
\label{eq_Abi_def}
\end{equation}
If $\mathbf{A}$ is invertible, then we have a unique solution. Otherwise, the existence of a solution depends on the nature of the system of equations. If the same is consistent, then all solutions can be expressed as
\begin{equation}
\mathbf{a}_i = \mathbf{A}^\dagger \mathbf{b}_i + \left(\mathbf{I} -  \mathbf{A}^\dagger \mathbf{A} \right) \mathbf{w},
\label{eq_aigen}
\end{equation}
where $\mathbf{w}$ is an arbitrary vector. In the next theorem, we prove that the equation always results in a solution. 

\begin{theorem}
For an FIR uniform analysis bank, there will be at least one optimal synthesis bank solution that minimizes the mean-square error. 
\end{theorem}
\begin{proof}
We have to prove that the solution given by~(\ref{eq_aigen}) is always an exact solution of the equation. Substituting the solution in the equation, we obtain
\begin{equation}
\begin{split}
\mathbf{A} \mathbf{a}_i &= \mathbf{A}\mathbf{A}^\dagger \mathbf{b}_i + \mathbf{A} \left(\mathbf{I} -  \mathbf{A}^\dagger \mathbf{A} \right) \mathbf{w}, \\
&= \mathbf{A}\mathbf{A}^\dagger \mathbf{b}_i.
\end{split}
\end{equation}
The matrix $\mathbf{A}$ can be expressed as
\begin{equation}
\begin{split}
\mathbf{A} &= \left( \mathbf{K} \mathbf{U} \mathbf{\Lambda} \mathbf{U}^H \mathbf{K}^H \right) ^T, \\
&= \left(\mathbf{K} \mathbf{U} \mathbf{\Lambda}^{1/2} \right)^* \left(\mathbf{K} \mathbf{U} \mathbf{\Lambda}^{1/2} \right)^T, \\
&= \mathbf{C} \mathbf{C}^H ,
\end{split}
\end{equation}
by performing eigendecomposition of the auto-correlation matrix $\mathbf{R}_{uu}$. With this,
\begin{equation}
\begin{split}
\mathbf{A} \mathbf{a}_i &=  \mathbf{C} \mathbf{C}^H (\mathbf{C} \mathbf{C}^H)^\dagger \mathbf{b}_i, \\
&= \mathbf{C} \mathbf{C}^H (\mathbf{C}^H)^\dagger \mathbf{C}^\dagger \mathbf{b}_i. \\
\end{split}
\end{equation}
Using the pseudoinverse identities~\citep{jn:ps-tutorial-review}, the equation reduces to 
\begin{equation}
\begin{split}
\mathbf{A} \mathbf{a}_i &= \mathbf{C} \mathbf{C}^\dagger \mathbf{b}_i, \\
&= (\mathbf{K} \mathbf{U} \mathbf{\Lambda}^{1/2} )^* \left( (\mathbf{K} \mathbf{U} \mathbf{\Lambda}^{1/2} )^*\right)^\dagger   \mathbf{K}^* (\mathbf{r}_{uu}^{i+d} )^T.
\end{split}
\end{equation}
The auto-correlation matrix $\mathbf{R}_{uu}$ is in general positive definite but always positive semi-definite. We consider the solution of equation under these two cases:
\begin{enumerate}
\item{Case I: Positive Definite Correlation Matrix}

In this case, the auto-correlation matrix is invertible, so is the diagonal matrix $\mathbf{\Lambda}$. Hence we can process the equation as follows:
\begin{equation}
\medmuskip=0mu
\begin{split}
\mathbf{A} \mathbf{a}_i &= (\mathbf{K} \mathbf{U} \mathbf{\Lambda}^{1/2} )^* ( (\mathbf{K} \mathbf{U} \mathbf{\Lambda}^{1/2} )^*)^\dagger   (\mathbf{K} \mathbf{U} \mathbf{\Lambda}^{1/2} )^*  \\
& \qquad \cdot \left(( \mathbf{U} \mathbf{\Lambda}^{1/2} )^*\right)^{-1} (\mathbf{r}_{uu}^{i+d} )^T, \\
&= (\mathbf{K} \mathbf{U} \mathbf{\Lambda}^{1/2} )^*  \left(( \mathbf{U} \mathbf{\Lambda}^{1/2} )^*\right)^{-1} (\mathbf{r}_{uu}^{i+d} )^T, \\
&= \mathbf{b}_i.
\end{split}
\end{equation}
\item{Case II: Positive Semi-Definite Correlation Matrix}

The diagonal matrix $\mathbf{\Lambda}$ is not invertible under this case. However we exploit the relation between the vector $\mathbf{r}_{uu}^{i+d}$ and the matrix $\mathbf{R}_{uu}$. For $(i+d) < (M(P-1)+Q)$, the vector $\mathbf{r}^{i+d}_{uu}$ is one of  the row of the correlation matrix. If $(i+d) \geq (M(P-1)+Q)$, then the vector $\mathbf{r}^{i+d}_{uu}$ can still be expressed as a linear combination of the rows of the auto-correlation matrix~\citep{bk:koch}. Thus, we can always express the vector $\mathbf{r}^{i+d}_{uu}$ as 
\begin{equation}
\mathbf{r}^{i+d}_{uu} = \sum_j c_j \mathbf{R}_{uu}(j,:) = \sum_j c_j \mathbf{U}(j,:) \mathbf{\Lambda} \mathbf{U}^H,
\end{equation}
where $\mathbf{R}_{uu}(j,:)$ is the $j$-th row of the correlation matrix. With this, 
\begin{equation}
\thickmuskip = 2mu
\begin{split}
\mathbf{A} \mathbf{a}_i &= (\mathbf{K} \mathbf{U} \mathbf{\Lambda}^{1/2} )^* \big( (\mathbf{K} \mathbf{U} \mathbf{\Lambda}^{1/2} )^*\big)^\dagger    (\mathbf{K} \mathbf{U} \mathbf{\Lambda}^{1/2} )^*  \\
& \qquad \cdot \mathbf{\Lambda}^{1/2} \big(\sum_j c_j \mathbf{U}(j,:)  \big)^T ,\\
&= (\mathbf{K} \mathbf{U} \mathbf{\Lambda}^{1/2} )^* \mathbf{\Lambda}^{1/2} \big(\sum_j c_j \mathbf{U}(j,:) \big)^T ,\\
&= \mathbf{b}_i.
\end{split}
\end{equation}
\end{enumerate} 
We can observe that we have 
\begin{equation}
\mathbf{A} \mathbf{a}_i = \mathbf{A}\mathbf{A}^\dagger \mathbf{b}_i = \mathbf{b}_i
\label{eq_thm1res1}
\end{equation}
under both the cases and hence the solution given by~(\ref{eq_aigen}) is exact. 
\end{proof}

From the above theorem, we can conclude that if the matrix $\mathbf{A}$ is non-invertible, we have infinite number of optimum solutions possible. This gives an opportunity to optimize the synthesis bank for a desired property. 

\subsubsection{Computational Complexity}
The computational complexity can be evaluated from various algebraic manipulations required for the solution in~(\ref{eq_aigen}). To start with, the matrix $\mathbf{A}$ evaluation needs $O\big(LP(M(P-1)+Q)^2\big) + O\big((LP)^2\big)$ calculations and its pseudoinverse requires $O\big((LP)^3\big)$ calculations. The term $(\mathbf{I} - \mathbf{A}^\dagger \mathbf{A})$ can then be evaluated with $O\big((LP)^3\big)$ cost. As the vector $\mathbf{b}_i$ needs to be evaluated for each $i \in [0, M-1]$, the total cost of obtaining $\mathbf{b}_i$'s is $O\big(MLP(M(P-1)+Q)\big)$. All the products $\mathbf{A}^\dagger \mathbf{b}_i$ and $(\mathbf{I} - \mathbf{A}^\dagger \mathbf{A})\mathbf{w}$ can then be obtained with complexity $O\big(M(LP)^2\big)$. Therefore, the complexity is quadratic with respect to $M$ and $Q$, and cubic with respect to $L$ and $P$. 

Our synthesis filter is multiple-input multiple-output (MIMO) FIR system, thus enabling the operations at the lowest sampling rate. The MIMO implementation brings down the complexity drastically, thereby for each reconstructed output sample, the number of multiplication and additions required are only $LP$ and $(LP-1)$, respectively.

\begin{figure*}[t]
\centering
\subfloat[Blocking of a subband]{\includegraphics[width = 0.5\linewidth]{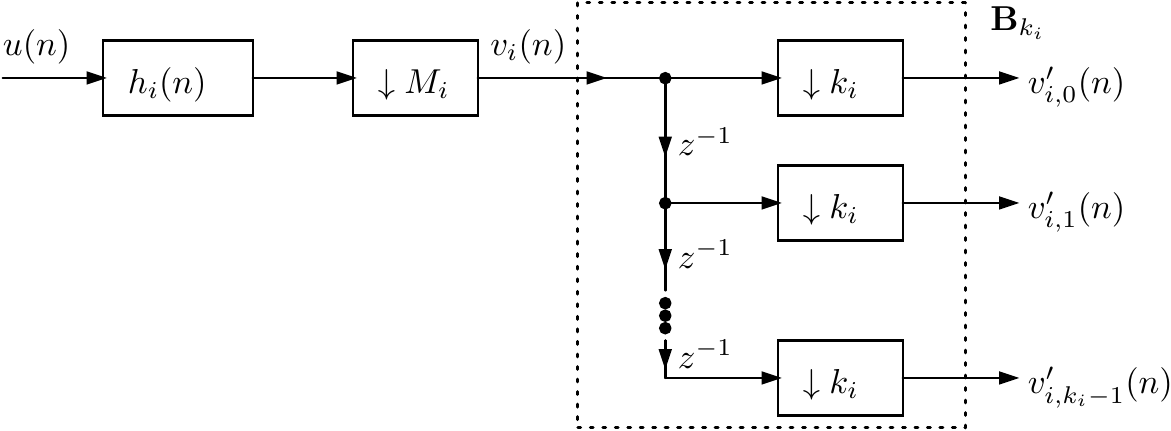}
\label{fig_nufb_ufb_block}} 
\subfloat[Equivalent representation]{\includegraphics[width = 0.38\linewidth]{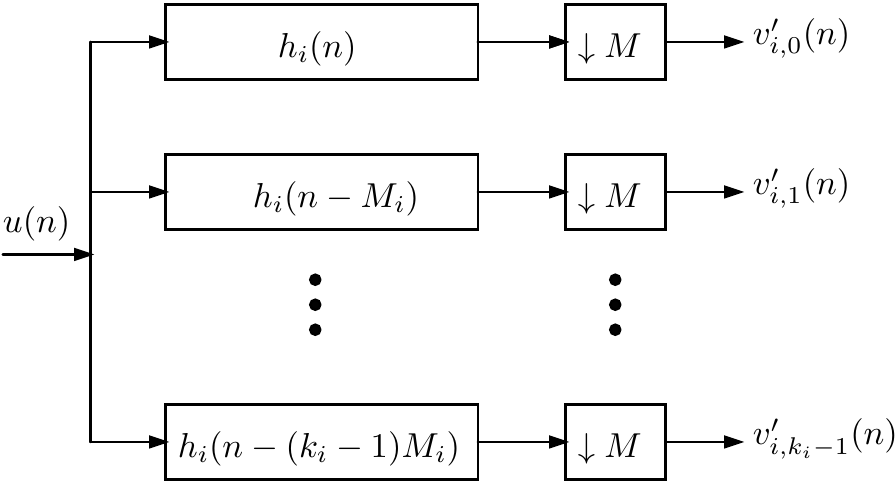}%
\label{fig_nufb_ufb_eqvblock}}
\caption{Conversion of an NUFB into an equivalent UFB}
\label{fig_nufb_to_ufb}
\end{figure*}

\subsection{Extension to NUFB}
The results developed so far for a UFB can be extended for an NUFB case. Before applying Wiener filtering on an NUFB, we convert it into an equivalent UFB with the help of blocking operation. This conversion is needed as Wiener filtering requires joint stationarity of the input and the desired signals~\citep{jn:sanelsevier}. Each subband of the NUFB with a decimation factor $M_i$ is blocked with a factor
\begin{equation}
k_i = M/M_i,\qquad 0 \leq i < L
\end{equation}
where $L$ is the number of subbands and $M$ is the LCM of all decimation factors. The process is shown in Fig.~\ref{fig_nufb_ufb_block}. An equivalent representation of the process, as shown in Fig.~\ref{fig_nufb_ufb_eqvblock}, can be derived using the noble identities~\citep{jn:sanelsevier,bk:vaidyn}. Thus with blocking, the NUFB behaves like a UFB with decimation factor equal to $M$ and with analysis filters given by
\begin{equation}
h'_{i,l}(n) = h_i(n-lM_i), \qquad 0 \leq i < L; \qquad 0 \leq l < k_i.
\end{equation}
Now the Wiener synthesis filter for the NUFB can be obtained using~(\ref{eq_aigen}).

\subsection{Minimum Mean-square Error}
\label{sec_mse}
We now determine the MSE for a given length of the Wiener filter. The result can be used to choose the length of the synthesis filter based on the desired level of the error. The MSE is defined as
\begin{equation}
\begin{split}
J &= \mathbf{E}[\mathbf{e}^H(n) \mathbf{e}(n)] = \sum_{i=0}^{M-1} \mathbf{E}[e_i^*(n) e_i(n)] .
\end{split}
\end{equation}
Its minimum value is achieved using the Wiener filter and is denoted by $J_{\text{min}}$. Further, $J_{\text{min}}$ is the sum of channel-wise MSEs which are given by 
\begin{equation}
\begin{split}
J_{\text{min}}^i &= \mathbf{E}[e_i^*(n) e_i(n)], \\
&= \mathbf{E}[e_i^*(n) d_i(n)] - \mathbf{E}[e_i^*(n) y_i(n)] .
\end{split}
\end{equation}
Using the principle of orthogonality~\cite{bk:haykin}, the above equation reduces to 
\begin{equation}
\begin{split}
J_{\text{min}}^i &= \mathbf{E}[e_i^*(n) d_i(n)], \\
&= \mathbf{E}[d_i^*(n) d_i(n)] - \mathbf{E}[y_i^*(n) d_i(n)], \\
&= r_{uu}(0) - \mathbf{E}[\mathbf{v}_s^H(n)\mathbf{a}_i^* d_i(n)], \\
&= r_{uu}(0) - \mathbf{r}_{d_i v} \mathbf{a}_i^* .
\end{split}
\end{equation}
Substituting the expressions for $\mathbf{r}_{d_i v}$ and $\mathbf{a}_i$ from~(\ref{eq_rdiv}) and~(\ref{eq_aigen}), respectively, we obtain
\begin{equation}
\medmuskip=0mu
\begin{split}
J_{\text{min}}^i &= r_{uu}(0) - \mathbf{b}_i^T \mathbf{A}^{\dagger*} \mathbf{b}_i^* - \left( \mathbf{b}_i - \mathbf{A}^H (\mathbf{A}^\dagger )^H \mathbf{b}_i  \right)^T \mathbf{w}^*, \\
&= r_{uu}(0) - \mathbf{b}_i^T \mathbf{A}^{\dagger*} \mathbf{b}_i^* - \left( \mathbf{b}_i - \mathbf{A} \mathbf{A}^\dagger  \mathbf{b}_i  \right)^T \mathbf{w}^*.
\end{split}
\end{equation}  
The term with the arbitrary vector is zero due to~(\ref{eq_thm1res1}) and the equation reduces to 
\begin{equation}
\begin{split}
J_{\text{min}}^i &= r_{uu}(0) - \mathbf{b}_i^T \mathbf{A}^{\dagger*} \mathbf{b}_i^*.
\label{eq_jmin_ch}
\end{split}
\end{equation}
This proves that all the solutions for the synthesis bank results in the same MSE, which is intuitive. 

\subsubsection{Optimal Delay}
The channel MSE $J_{\text{min}}^i$ depends on the delay $d$. The same is also observed in our simulation results. The exact form of the relation can be obtained from~(\ref{eq_Abi_def}) as
\begin{equation}
\begin{split}
J_{\text{min}}^i &= r_{uu}(0) - \mathbf{r}_{uu}^{i+d} \mathbf{K}^H  \left( \mathbf{K} \mathbf{R}_{uu} \mathbf{K}^H \right) ^\dagger \mathbf{K} \left(\mathbf{r}_{uu}^{i+d}\right)^H \\
&= r_{uu}(0) - \mathbf{r}_{uu}^{i+d} \mathbf{B} \left(\mathbf{r}_{uu}^{i+d}\right)^H
\end{split}
\end{equation}
$J_{\text{min}}^i$ can be observed as a quadratic form with the matrix $\mathbf{B}$ being a Hermitian matrix. If we are interested to obtain an optimal delay value, we can use the existing results~\citep{bk:moon2000mathematical,jn:quad_opt} on the quadratic forms.      

In the following section, we derive the necessary and sufficient conditions on a filter bank for the MSE to be zero. 

\section{FIR Perfect Reconstruction Filter Bank}
\label{sec_firpr}
In the following analysis, we prove that the designed synthesis bank can achieve PR provided that the analysis bank satisfies certain condition. To justify PR property for any arbitrary input signal on a strong analytical background, we verify the pseudocirculant property~\citep{bk:vaidyn,jn:vaidypoly} for the solution. 

\subsection{Pseudocirculant Property}
We evaluate the product $\mathbf{P}(z) = \mathbf{A}(z)\mathbf{E}(z)$ in time-domain where $\mathbf{A}(z)$ and $\mathbf{E}(z)$ are the polyphase component matrices for the synthesis bank and the analysis bank, respectively~\citep{jn:vaidypoly}. The matrix $\textbf{P}(z)$ is designed as an identity matrix in previous works~\citep{jn:duval,jn:Cvetkovic}. Our design approach is more general where $\textbf{P}(z)$ satisfies the pseudocirculant property. In time-domain, we obtain the matrix $\textbf{P}(z)$ from convolution of the two matrix filters as
\begin{equation}
p_{i,j}(n) = \sum_{r=0}^{L-1} a_{i,r}(n) \ast e_{r,j}(n), \qquad 0 \leq i,j \leq M-1 .
\label{eq_pijt1}
\end{equation}
In the above equation, the filter $e_{r,j}(n)$ can be obtained from the analysis filters as~\citep{bk:vaidyn}
\begin{equation}
e_{r,j}(n) = h_r(Mn+j).
\end{equation}
As $Q$ is the length of the analysis filters, the length of the filter $e_{r,j}(n)$ will be at max
\begin{equation}
T = \left\lceil \frac{Q}{M}\right\rceil.
\end{equation}
By zero padding, we ensure that all the constituent filters of the polyphase filter $\mathbf{E}(n)$ have the same length $T$. With that, the convolution given in~(\ref{eq_pijt1}) can be expressed in matrix form using~(\ref{eq_conv_form2}) as
\begin{equation}
\mathbf{p}_{i,j} = \sum_{r=0}^{L-1} \tilde{\mathbf{E}}_{r,j} \mathbf{a}_{i,r} ,
\label{eq_pijt2}
\end{equation}
where $\tilde{\mathbf{E}}_{r,j}$ is a convolution matrix of size $(P+T-1) \times P$. We like to express this convolution matrix in terms of the analysis filter $h_r(n)$. For that, we use a polyphase identity, as shown in Fig.~\ref{fig_poly_indentity}, obtained from the existing result~\citep{bk:vaidyn} by finding the $0$-th polyphase component of a time-advanced filter $h_r(n+j)$.
\begin{figure*}[t]
\centering
\subfloat[A polyphase component]{\includegraphics{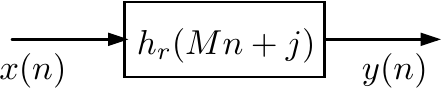}%
\label{fig_poly_indentity1} }
\qquad \subfloat[Equivalent multirate circuit]{\includegraphics{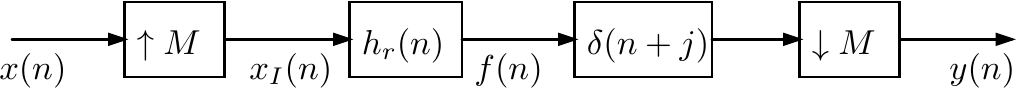}%
\label{fig_poly_indentity2}}
\caption{The polyphase identity.}
\label{fig_poly_indentity}
\vspace{-12pt}
\end{figure*}
Using this identity, we obtain the following relation
\begin{equation}
\begin{split}
\tilde{\mathbf{E}}_{r,j} &= \mathbf{D}_{P+T-1 , M(P-1)+Q}^j \tilde{\mathbf{H}}_r (\mathbf{D}_{P , M(P-1)+1}^0)^T, \\
&= \mathbf{D}_{P+T-1 , M(P-1)+Q}^j \overbar{\mathbf{H}}_r^T (\mathbf{D}_{P , M(P-1)+1}^0)^T, \\
&= \mathbf{D}_{P+T-1 , M(P-1)+Q}^j \left( \mathbf{D}_{P , M(P-1)+1}^0 \overbar{\mathbf{H}}_r \right)^T .
\end{split}
\end{equation}
It can be seen that Fig.~\ref{fig_poly_indentity2} contains a combination of upsampling, filtering, time-advance and downsampling. Even then the matrix representations developed in this paper for various multirate operations have been successfully applied to deal with it. Continuing forward, we use the above relation in~(\ref{eq_pijt2}) and obtain
\begin{align}
\mathbf{p}_{i,j} &= \mathbf{D}_{P+T-1 , M(P-1)+Q}^j \sum_{r=0}^{L-1} \left( \mathbf{D}_{P , M(P-1)+1}^0 \overbar{\mathbf{H}}_r \right)^T \mathbf{a}_{i,r}, \notag \\
&= \mathbf{D}_{P+T-1 , M(P-1)+Q}^j \mathbf{K}^T \mathbf{a}_i .
\label{eq_pijtimed}
\end{align}
Next we derive a necessary and sufficient condition, in the form of a theorem, for the analysis and synthesis banks of a filter bank to result in PR.

\begin{theorem}
An FIR uniform filter bank is PR if and only if
\begin{equation}
\mathbf{K}^T \mathbf{a}_i = c \mathbf{e}_{i+d},
\end{equation}
where $c$ is a non-zero constant and the column vector  $\mathbf{e}_{i+d}$ is defined as follows:
\begin{equation}
{e}_{i+d} (l) = \begin{cases}
1, \qquad &\text{if } l = i+d \\
0, \qquad &\text{otherwise.}
\end{cases} 
\end{equation}
\end{theorem}
\begin{proof}
The proof follows the Vaidyanathan's well-known pseudocirculant property. A PR filter bank must satisfies the generalized pseudocirculant condition given by~\citep{bk:vaidyn}
\begin{equation}
\mathbf{P}(z) = c z^{-n_0} \begin{bmatrix}
\mathbf{0} & \mathbf{I}_{M-r} \\
z^{-1} \mathbf{I}_r & \mathbf{0}
\end{bmatrix}.
\end{equation}
We have evaluated $\mathbf{P}(z)$ in time-domain and next we find conditions for it to be pseudocirculant. Let $\mathbf{K}^T \mathbf{a}_i = \mathbf{s}_i$, then~(\ref{eq_pijtimed}) becomes
\begin{equation}
\medmuskip=0mu
\arraycolsep=2pt
\mathbf{p}_{i,j} = \begin{bmatrix}
s_i(j) & s_i(j+M) & \dots & s_i(j+M(P+T-2))
\end{bmatrix}^T.
\end{equation}
If the matrix filter $\mathbf{P}(z)$ is pseudocirculant, then the following properties are true in time-domain:
\begin{enumerate}
\item Property 1: In a row, only one $\mathbf{p}_{i,j}$ is non-zero i.e. only one out of $\mathbf{p}_{i,0}, \allowbreak \mathbf{p}_{i,1}, \allowbreak \dots, \allowbreak \mathbf{p}_{i,M-1}$ is non-zero.
\item Property 2: In a column, only one $\mathbf{p}_{i,j}$ is non-zero i.e. only one out of $\mathbf{p}_{0,j}, \allowbreak \mathbf{p}_{1,j}, \allowbreak \dots, \allowbreak \mathbf{p}_{M-1,j}$ is non-zero.
\item Property 3: For a vector $\mathbf{p}_{i,j}$, at most one entry is non-zero.
\item Property 4: All the entries on upper-left to lower-right diagonals are the same i.e.
\begin{equation}
\mathbf{p}_{i+1,j+1} = \mathbf{p}_{i,j}, \qquad 0 \leq i,j \leq M-2. 
\end{equation}
\item Property 5: The first column will have 
\begin{equation}
\medmuskip=0mu
\begin{split}
p_{i+1,0}(l) &= \begin{cases}
0, \qquad & \text{if } l = 0 \\
p_{i,M-1}(l-1), \qquad & \text{otherwise}
\end{cases}, \\
& \qquad 0 \leq i \leq M-2.
\end{split}
\end{equation}
\end{enumerate}
The above properties will result in the following constraints on the vector $\mathbf{s}_i$:
\begin{enumerate} 
\item Property 1 implies that only one vector 
\begin{equation}
\thickmuskip = 2mu
\medmuskip=0mu
\arraycolsep=1pt
\begin{split}
\mathbf{p}_{i,j} &= 
\begin{bmatrix}
s_i(j) & s_i(j+M) & \dots & s_i(j+M(P+T-2)) 
\end{bmatrix}^T, \\
& \qquad 0 \leq j < M
\end{split}
\end{equation}
is non-zero in a row. Further, Property 3 means that only one element of the non-zero vector is non-zero. Together these two conditions imply that only one element of the vector $\mathbf{s}_i$ is non-zero. Assume that $\mathbf{p}_{i,j} \neq 0$ for $j = j_i$ and
${p}_{i,j_i}(l_i) = s_i(j_i + Ml_i) \neq 0$, then we can express
\begin{equation}
\mathbf{s}_i = c_i \mathbf{e}_{j_i + Ml_i}.
\label{eq_si0}
\end{equation}
\item Property 4 results in
\begin{equation}
\thickmuskip = 2mu
\medmuskip=0mu
\begin{split}
s_{i+1}(Ml+j+1) &= s_i(Ml+j), \qquad 0 \leq j \leq M-2; \\
& \qquad 0 \leq l \leq P+T-2.
\end{split}
\label{eq_si1}
\end{equation}
From the above equation, all the elements of the vector $\mathbf{s}_{i+1}$, except at indices that are multiples of $M$, 
can be obtained from the vector $\mathbf{s}_i$.
\item Property 5 implies 
\begin{equation}
p_{i+1,0}(0) = s_{i+1}(0)  = 0, \qquad 0 \leq i \leq M-2.
\end{equation}
Another implication is that 
\begin{equation}
s_{i+1}(Ml) = s_i(Ml-1). 
\label{eq_si2}
\end{equation}
Combining the two results~(\ref{eq_si1}) and~(\ref{eq_si2}), we get
\begin{equation}
s_{i+1}(l) = \begin{cases}
0, & \qquad l = 0, \\
s_i(l-1), & \qquad \text{otherwise.}
\end{cases}
\end{equation}
Thus the vector $\mathbf{s}_{i+1}$ is just a shifted version of the vector $\mathbf{s}_{i}$. 
\end{enumerate}
The last constraint implies that all the vectors $\mathbf{s}_i$ can be obtained from a single vector  $\mathbf{s}_0$. If the vector $\mathbf{s}_0$ has form  
\begin{equation}
\mathbf{s}_0 = c \mathbf{e}_{Ml_0 + j_0} = c \mathbf{e}_{d_0},
\end{equation} 
then all the vectors $\mathbf{s}_i$ can be expressed as
\begin{equation}
\mathbf{s}_i = c \mathbf{e}_{i+d_0}, \qquad 0 \leq i \leq M-1.
\end{equation}
Due to the above equation, we have
\begin{equation}
\mathbf{p}_{i,j} = c \mathbf{D}_{P+T-1 , M(P-1)+Q}^j \mathbf{e}_{i+d_0} .
\end{equation}
In z-domain, this means that
\begin{equation}
\mathbf{P}(z) = c z^{-l_0} \begin{bmatrix}
\mathbf{0} & \mathbf{I}_{M-j_0} \\
z^{-1} \mathbf{I}_{j_0} & \mathbf{0}
\end{bmatrix}
\end{equation}
and the reconstructed output is $\hat{u}(n) = cu(n-d_0-M+1)$. One can easily see that the variable $d_0$ is the same as the delay $d$ in the desired signal path.
\end{proof}
It is to be noted that the above theorem puts restriction on both the analysis and synthesis banks. It may happen that a solution may not exist for the synthesis bank to result in PR. We will first obtain a parameterization of all PR solutions for the synthesis bank and later derive necessary \& sufficient conditions for the same to exist. 

\subsection{Parametric Form for PR Solutions}
In the previous subsection, we obtained the following condition on the analysis and synthesis banks to have PR:
\begin{equation}
\mathbf{K}^T \mathbf{a}_i = c \mathbf{e}_{i+d}, \qquad 0 \leq i < M.
\end{equation}
The above equation is of form 
\begin{equation}
\mathbf{A} \mathbf{a}_i = \mathbf{b}_i
\end{equation}
and thus represents a system of equations. If the same is consistent, then all the PR solutions for the synthesis bank are given by
\begin{equation}
\begin{split}
\mathbf{a}_i &= \mathbf{A}^\dagger \mathbf{b}_i + (\mathbf{I} - \mathbf{A}^\dagger \mathbf{A})\mathbf{w}, \\
 &= c(\mathbf{K}^\dagger)^T \mathbf{e}_{i+d} + (\mathbf{I} - \mathbf{K} \mathbf{K}^\dagger)^T \mathbf{w},
\end{split}
\label{eq_genPR}
\end{equation}
where $\mathbf{w}$ is an arbitrary vector. We can observe that the solutions have turned out to be independent of the input signal. 

\subsection{Necessary and Sufficient Conditions on Analysis Bank and Delay}
In order for a PR solution to exist, the analysis bank has to satisfy certain condition. We present the same in the form of a theorem about the matrix $\mathbf{K}$ which represents the analysis bank.
\begin{theorem} \label{thm_analysis_pr}
For an FIR uniform analysis bank, at least one PR synthesis bank solution is possible if the matrix $\mathbf{K}$ satisfies the relation
\begin{equation}
\medmuskip=0mu
\thickmuskip = 0mu
\arraycolsep=1pt
\resizebox{0.88\linewidth}{!}{$
\mathbf{K}^\dagger \mathbf{K} = \mathbf{I} + \begin{bmatrix}
\mathbf{B}_{p \times p} & \mathbf{0}_{p \times r} & \mathbf{C}_{p \times (q-p-r)} \\
\mathbf{0}_{r \times p} & \mathbf{0}_{r \times r} & \mathbf{0}_{r \times (q-p-r)} \\
\mathbf{C}^H_{p \times (q-p-r)} & \mathbf{0}_{(q-p-r) \times r} & \mathbf{D}_{(q-p-r) \times (q-p-r)}  
\end{bmatrix},$}
\label{eq_k_pr_cond}
\end{equation}
where $\mathbf{0}_{r \times r}$ is the largest zero sub-block possible with $r \geq M$, and $\mathbf{B}_{p \times p}$, $\mathbf{D}_{(q-p-r) \times (q-p-r)}$ and $\mathbf{C}_{p \times (q-p-r)}$ are arbitrary matrices with the first two being Hermitian.
\end{theorem}
\begin{proof}   
We require the system of linear equations given by 
\begin{equation}
\mathbf{K}^T \mathbf{a}_i = c \mathbf{e}_{i+d}
\end{equation}
to be consistent to have an exact solution. In order to obtain the condition for that, we substitute a PR solution in the above equation and obtain 
\begin{equation}
\begin{split}
c \left( \mathbf{K}^\dagger \mathbf{K} \right)^T \mathbf{e}_{i+d} + \left( \mathbf{K} - \mathbf{K} \mathbf{K}^\dagger \mathbf{K} \right)^T \mathbf{w} &= c \mathbf{e}_{i+d}. \\
\end{split}
\end{equation}
Using the fact that $\mathbf{K} \mathbf{K}^\dagger \mathbf{K} = \mathbf{K}$, the equation simplifies to
\begin{equation}
\begin{split}
\left( \left( \mathbf{K}^\dagger \mathbf{K}\right)^T - \mathbf{I}\right) \mathbf{e}_{i+d} &= \mathbf{0}, \qquad 0 \leq i \leq M-1 .
\end{split}
\label{eq_pr_ab_temp1}
\end{equation}
If the matrix $\mathbf{K}^\dagger \mathbf{K}$ is of dimension $q \times q $, then in order for the above equation to be satisfied, we require
\begin{equation}
\left(\mathbf{K}^\dagger \mathbf{K} \right)^H = \mathbf{I} + \begin{bmatrix}
\mathbf{X}_{q \times p} & \mathbf{0}_{q \times r} & \mathbf{Y}_{q \times (q-p-r)} 
\end{bmatrix},
\label{eq_pr_ab_temp2} 
\end{equation}
where $p \leq d$, $r \geq M$, and $\mathbf{X}_{q \times p}$ and $\mathbf{Y}_{q \times (q-p-r)}$ are arbitrary matrices of dimensions $ q \times p$ and $q \times (q-p-r)$, respectively. As $\mathbf{K}^\dagger \mathbf{K}$ is a Hermitian matrix, we have
\begin{equation}
\begin{bmatrix}
\mathbf{X}_{q \times p} & \mathbf{0}_{q \times r} & \mathbf{Y}_{q \times (q-p-r)} 
\end{bmatrix} = \begin{bmatrix}
\mathbf{X}^H_{q \times p} \\
\mathbf{0}_{r \times q} \\
\mathbf{Y}^H_{q \times (q-p-r)} 
\end{bmatrix}.
\end{equation}
Thus the two matrices $\mathbf{X}_{q \times p}$ and $\mathbf{Y}_{q \times (q-p-r)}$ have following form:
\begin{equation}
\mathbf{X}_{q \times p} = \begin{bmatrix}
\mathbf{B}_{p \times p} \\
\mathbf{0}_{r \times p} \\
\mathbf{C}^H_{p \times (q-p-r)}
\end{bmatrix}
\end{equation}
and 
\begin{equation}
\mathbf{Y}_{q \times (q-p-r)} = \begin{bmatrix}
\mathbf{C}_{p \times (q-p-r)} \\
\mathbf{0}_{r \times (q-p-r)} \\
\mathbf{D}_{(q-p-r) \times (q-p-r)}
\end{bmatrix},
\end{equation}
where matrices $\mathbf{B}_{p \times p}$ and $\mathbf{D}_{(q-p-r) \times (q-p-r)}$ are Hermitian matrices. Hence the condition on the matrix $\mathbf{K}$ becomes
\begin{equation}
\arraycolsep=1pt
\thickmuskip = 0mu
\medmuskip=0mu
\resizebox{0.88\linewidth}{!}{$
\mathbf{K}^\dagger \mathbf{K} = \mathbf{I} + \begin{bmatrix}
\mathbf{B}_{p \times p} & \mathbf{0}_{p \times r} & \mathbf{C}_{p \times (q-p-r)} \\
\mathbf{0}_{r \times p} & \mathbf{0}_{r \times r} & \mathbf{0}_{r \times (q-p-r)} \\
\mathbf{C}^H_{p \times (q-p-r)} & \mathbf{0}_{(q-p-r) \times r} & \mathbf{D}_{(q-p-r) \times (q-p-r)}  
\end{bmatrix}.$}
\end{equation}
\end{proof}
The PR requirement imposes the following condition on the desired signal delay $d$. 
\begin{lemma}
With an analysis bank matrix $\mathbf{K}$ satisfying the above condition, PR is possible only when the delay $d$ is in the range: 
\begin{equation}
p \leq d \leq p+r-M.
\label{eq_delay_prcondition}
\end{equation}
\end{lemma}
\begin{proof}
The proof easily follows from the equations~(\ref{eq_pr_ab_temp1}) and~(\ref{eq_pr_ab_temp2}). 
\end{proof}
The above lemma together with Theorem~\ref{thm_analysis_pr} constitute the necessary and sufficient conditions for a PR synthesis solution to exist. The lemma points that with suitable delay, PR is achievable for an analysis bank which is otherwise not possible. This way we have attempted to expand the space of PR FIR filter banks. Further, we can have PR for a number of delay values which gives an opportunity to use delay as an optimization parameter, based on various design criterion like Hermitian symmetry, time or frequency localization~\citep{jn:duval}, coding gain~\citep{jn:labeau}, etc. For this, the number of channels required is $L > M$ and the channels should satisfy the PR condition given by Theorem~\ref{thm_analysis_pr}. For a fixed delay, it is an unconstrained problem with respect to the arbitrary vector $\mathbf{w}$~\citep{jn:duval,jn:bolcskeiframe}. The required optimized solution can  be selected among the solutions for various delays by imposing an additional constraint. Design specific constraints and their solutions can be considered for future work. 

It can be seen that the matrix $\mathbf{K}$ easily satisfy the PR condition if it has full column rank. The PR condition induces certain restriction on its null space which is given by the following lemma. 
\begin{lemma}
For a PR UFB, the analysis bank matrix $\mathbf{K}$ has a null space spanned by the following set of vectors:
\begin{equation}
\mathbf{z} = \begin{bmatrix}
\mathbf{f}_{ 1 \times p} & \mathbf{0}_{1 \times r} & \mathbf{g}_{1 \times (q-p-r)} \\
\end{bmatrix}^T,
\end{equation}
where $\mathbf{f}_{ 1 \times p}$ and $\mathbf{g}_{1 \times (q-p-r)}$ are arbitrary vectors. 
\end{lemma}
\begin{proof} 
The matrix $\mathbf{K}$ represents the analysis bank of a filter bank. Hence, when an input signal $u(n)$ is applied to the analysis bank, the resulting subband vector can be obtained as
\begin{equation}
\mathbf{v}_s(n) = \mathbf{K} \bar{\mathbf{u}}(Mn).
\end{equation}
Consider a case where  the vector $\bar{\mathbf{u}}(Mn)$ lies in the null space of the matrix $\mathbf{K}$. As a result, the vector $\mathbf{v}_s(n)$ as well as the output of the synthesis bank, both are zero. The reconstructed signal is then 
\begin{equation}
\hat{u}(Mn+i) = y_{M-1-i}(n) = 0, \qquad 0 \leq i < M.
\end{equation}
For PR, we require
\begin{equation}
\hat{u}(n) = u(n-d-M+1). 
\end{equation}
Combining the two, we get
\begin{equation}
\thickmuskip = 2mu
\medmuskip = 0mu  
\hat{u}(Mn+i) = u(Mn+i-d-M+1) = 0 , \qquad 0 \leq i \leq M-1.
\end{equation}
Thus the vector $\bar{\mathbf{u}}(Mn)$ should be of form
\begin{equation}
\thickmuskip = 0mu
\medmuskip = 0mu  
\arraycolsep = 1pt
\bar{\mathbf{u}}(Mn) = \begin{array}[t]{cccccccc}
 \bigl[u(Mn) & u(Mn-1) & \dots & u(Mn-d+1) & \mathbf{0}_{1 \times M} \\
  & u(Mn-d-M) & \dots & u(Mn-q+1) \bigr]^T.
 \end{array}
\end{equation} 
A basis vector for the null space of the matrix $\mathbf{K}$ should satisfy the above condition for all possible values of $d$ and hence should have the form
\begin{equation}
\mathbf{z} = \begin{bmatrix}
\mathbf{f}_{ 1 \times p} & \mathbf{0}_{1 \times r} & \mathbf{g}_{1 \times (q-p-r)} \\
\end{bmatrix}^T,
\end{equation}
where $\mathbf{f}_{ 1 \times p}$ and $\mathbf{g}_{1 \times (q-p-r)}$ are arbitrary vectors. 
\end{proof}
In this section, we presented various necessary and sufficient conditions for a PR synthesis bank solution to exist. In the next section, we present some experiments which validates our results. 

\begin{figure*}[!t]
\centering
\subfloat[Channel 0]{\includegraphics[width=0.46\linewidth]{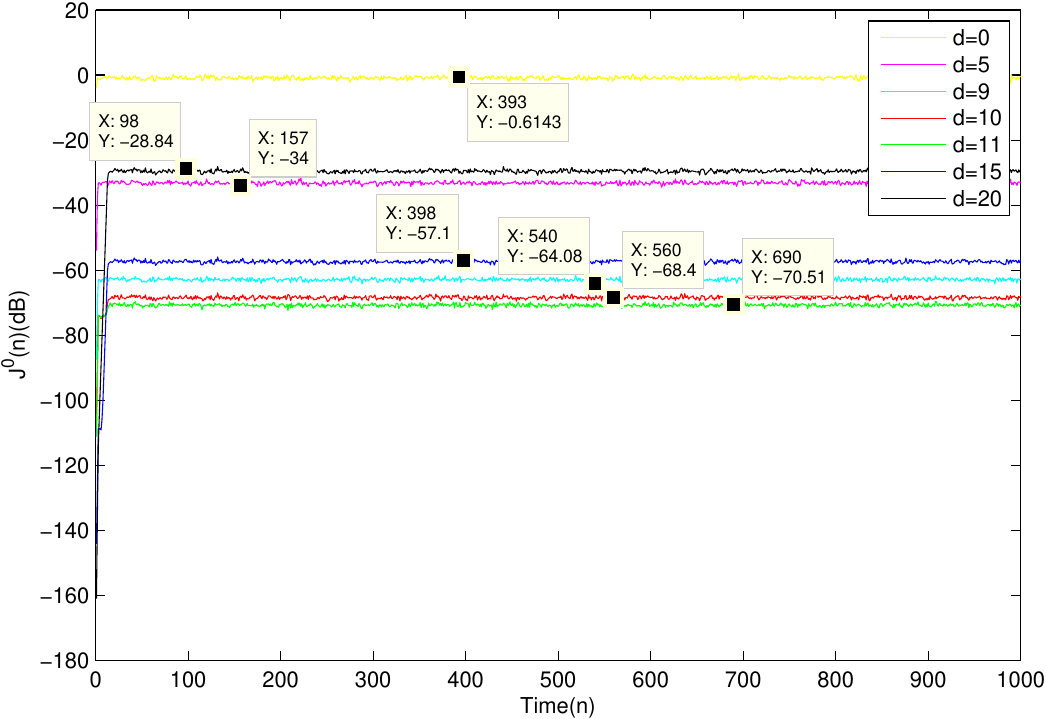}
\label{fig_exp_ensbch0}}
\subfloat[Channel 1]{\includegraphics[width=0.46\linewidth]{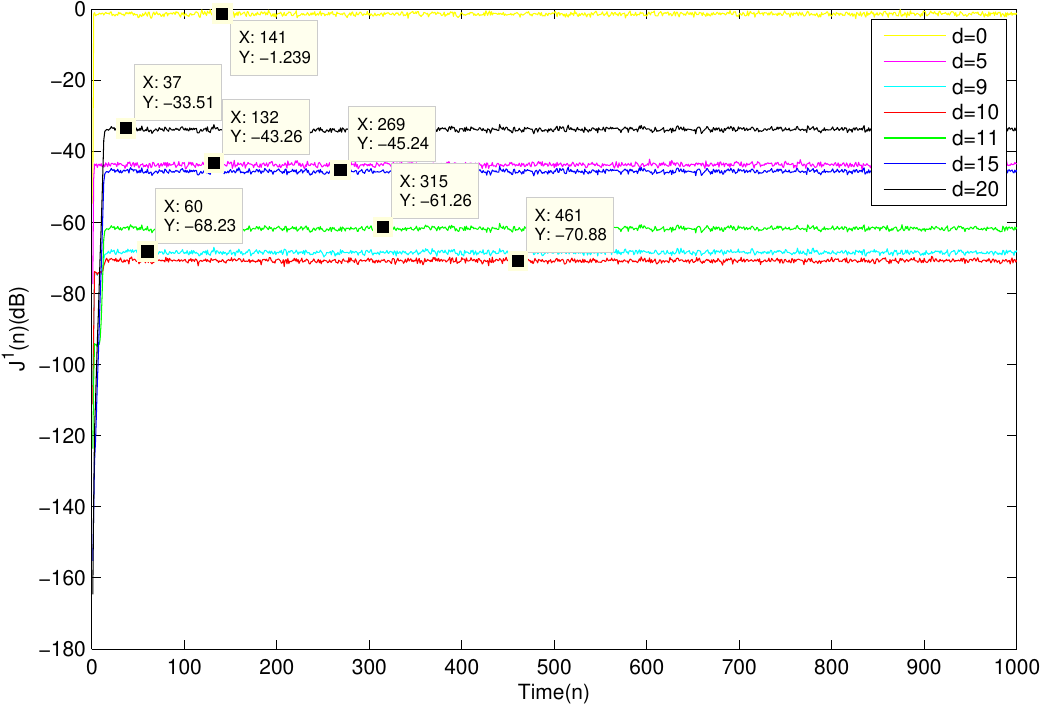}%
\label{fig_exp_ensbch1}}
\caption{Experiment 1: Ensemble-averaged squared error for various delay values}
\label{fig_exp_ensb}
\vspace{-12pt}
\end{figure*}
\section{Experimental Results}
\label{sec_exp}
In this section, we verify the derived results with a number of experiments. These experiments are chosen to demonstrate the wide scope of application of our design. In Experiment 1, we obtain the MSE for various delay values in the desired signal and compare the results. We present the improvement in reconstruction, in terms of the MSE, with addition of a new subband in Experiment 2. We also examine the impact of delay on PR. In Experiment 3, a PR non-uniform filter bank obtained through our design is verified with a non-stationary input signal. 
\subsection{Experiment 1}
We consider a 2-channel maximally decimated uniform analysis bank designed using fir1 routine of Matlab with the following parameters: the analysis filter $h_0(n)$ is low-pass filter with cut-off 0.6 and length 9, and $h_1(n)$ is high-pass filter with cut-off 0.4 and length 10. The input to the filter bank is a zero mean, unit variance second-order autoregressive(AR(2)) process with coefficients $0.7$ and $0.1$. We choose the length of the constituent filters of the matrix synthesis filter to be 11 and design the filter for various delay values using~(\ref{eq_aigen}). We obtain ensembled-averages of the channel-wise squared errors, shown in Figs.~\ref{fig_exp_ensbch0} and ~\ref{fig_exp_ensbch1}, for these synthesis filters by performing 200 independent runs. The corresponding theoretical values, obtained using~(\ref{eq_jmin_ch}), are provided in Table~\ref{tab_exp1}. It is to be mentioned here that the analysis bank does not satisfy Theorem~\ref{thm_analysis_pr}. We can see that the observed error values  are closer to the theoretical values for both channels. Further, we can notice that the best result among the selected delay values, in terms of the total MSE, is with delay 10. This shows that a better reconstruction can be obtained if some delay in the reconstructed output is acceptable. 
\begin{table}[htbp]
\caption{Experiment 1: MSE for various delay values}
\label{tab_exp1}
\vspace{-5pt}
\tabcolsep = 3pt
\centering
\resizebox{0.48\linewidth}{!}{
\begin{tabular}{c c c c }
\toprule 
{\centering $d$} & {\centering $J_{\text{min}}^0$(dB)} & {\centering $J_{\text{min}}^1$(dB)} & {\centering $J_{\text{min}}$(dB)} \\
\midrule
0 & -0.8673 & -1.3434 &  1.9115 \\ 
5 & -33.1717 & -43.7244 &  -32.8052 \\ 
9 &  -62.8689 & -68.3970 & -61.7967 \\ 
10 &  -68.3970 & -70.7269 &  -66.3972 \\ 
11 & -70.7269 & -61.6691 & -61.1606 \\ 
15 & -57.3320 & -45.5881 &  -45.3068 \\ 
20 &  -29.5288 & -33.7826 & -28.1442 \\
\bottomrule
\end{tabular}}
\vspace{-12pt}
\end{table}

\begin{table*}[!t]
\caption{Experiment 2: Variation in MSE with number of available subbands}
\label{tab_mse_vs_nosubband}
\vspace{-4pt}
\tabcolsep = 3pt
%\belowcaptionskip = -10pt
\centering
\resizebox{0.6\linewidth}{!}{
\begin{tabular}{c c c c c c c c}
\toprule 
{\centering $L$} & {\centering Available analysis filters} & {\centering $d$} & {\centering $J_{\text{min}}^0$(dB)} & {\centering $J_{\text{min}}^1$(dB)} & {\centering $J_{\text{min}}^2$(dB)} & {\centering $J_{\text{min}}^3$(dB)} & {\centering $J_{\text{min}}$(dB)} \\
\midrule
1 & $h_0(n)$ & 10 & -12.9106 & -12.9711 & -15.0069 & -14.9634 &  -7.8231 \\
1 & $h_1(n)$ & 10 &  -0.2141 &  -0.2360 &  -0.0490 &  -0.1161 & 5.8675 \\
1 & $h_2(n)$ & 10 & -0.0826 &  -0.0577 &  -0.1747 &  -0.0823 & 5.9215 \\
1 & $h_3(n)$ & 10 &  -0.0865  & -0.0499 &  -0.0612 &  -0.0241 & 5.9652 \\
2 & $h_0(n)$, $h_1(n)$ & 10 & -16.7631 & -19.0161 & -16.7292 & -19.1518 & -11.7388 \\
2 & $h_0(n)$, $h_2(n)$ & 10 &  -14.3354 & -13.1462 & -18.6208 & -15.8517 &  -9.0295 \\
2 & $h_0(n)$, $h_3(n)$ & 10 & -13.7442 & -13.8427 & -15.6931 & -15.6224 & -8.6055 \\
2 & $h_1(n)$, $h_2(n)$ & 10 & -0.2899 &  -0.2923 &  -0.2217 &  -0.2139 &  5.7663 \\
2 & $h_1(n)$, $h_3(n)$ & 10 & -0.3044 &  -0.2881 &  -0.1109 &  -0.1410 & 5.8104 \\
2 & $h_2(n)$, $h_3(n)$ & 10 &  -0.1676 &  -0.1071 &  -0.2355 &  -0.1088 & 5.8662 \\
3 & $h_0(n)$, $h_1(n)$, $h_2(n)$ & 10 &  -20.5040 & -20.3481 & -23.1342 & -23.1368 & -15.5518 \\
3 & $h_0(n)$, $h_1(n)$, $h_3(n)$ & 10 & -19.2065 & -24.3326 & -17.8184 & -21.2263 & -14.0058 \\
3 & $h_0(n)$, $h_2(n)$, $h_3(n)$ & 10 & -15.5108 & -14.0397 & -20.4628 & -16.7167 & -10.0944 \\
3 & $h_1(n)$, $h_2(n)$, $h_3(n)$ & 10 & -0.3786 &  -0.3432  & -0.2834  & -0.2416 & 5.7092 \\
4 & $h_0(n)$, $h_1(n)$, $h_2(n)$, $h_3(n)$ & 10 &  -45.1648 & -32.0844 & -145.2323 & -143.7481 & -31.8758\\
4 & $h_0(n)$, $h_1(n)$, $h_2(n)$, $h_3(n)$ & 11 &  -32.0844 & -145.2323 & -143.7481 & -143.5253 & -32.0844\\
4 & $h_0(n)$, $h_1(n)$, $h_2(n)$, $h_3(n)$ & 12 &  -145.2323 & -143.7481 & -143.5253 & -145.3962 & -138.3732\\
4 & $h_0(n)$, $h_1(n)$, $h_2(n)$, $h_3(n)$ & 13 &  -143.7481 & -143.5253 & -145.3962 & -30.2194 & -30.2194\\
\bottomrule
\end{tabular}}
\vspace{-12pt}
\end{table*}
\subsection{Experiment 2}
In this experiment, we examine the effect of number of subbands on the MSE along with the role of delay in achieving PR. We consider an Extended Lapped Transform (ELT) filter bank with decimation factor $M=4$ and overlapping factor $K'=2$. The prototype filter for ELT is designed using the optimized butterflies angles mentioned in~\citep{jn:malvar_fast_elt}. The input signal to the filter bank is a realization of an AR(1) process with coefficient $0.95$, zero mean and unit variance. The length of the synthesis filter is $4$ and the delay is $10$. We vary the number of subbands ($L$) available to the synthesis filter and obtain the MSE for each channel. On analyzing the obtained results, as given in Table~\ref{tab_mse_vs_nosubband}, we can conclude that for number of subbands $ L < 4$, there is a variation in the MSE with subbands chosen. However independent of the choice of subbands, if we add additional subband to the system, then we always see an improvement in the total MSE as well as in the channel-wise MSEs. This is possible since we have exploited the correlation between the input signal components while designing the matrix synthesis filter. With the number of subbands $L = 4$, the analysis bank satisfies the PR condition given by Theorem~\ref{thm_analysis_pr}. We vary the delay and find out that PR is possible only with $d=12$. This is consistent with the range of delay values provided by~(\ref{eq_delay_prcondition}).

\begin{figure}[htbp]
\centering
\includegraphics[scale = 0.55]{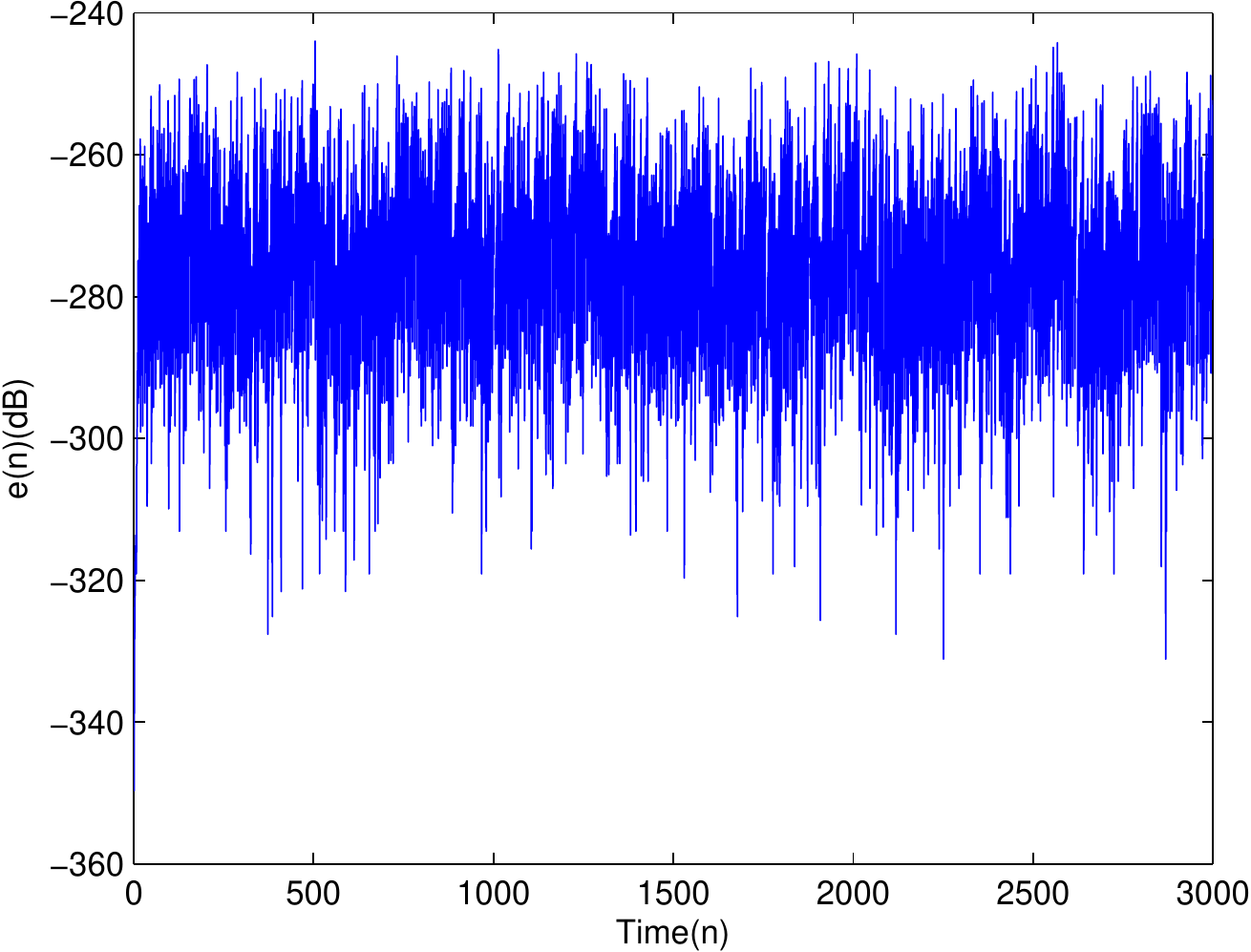}
\vspace{-8pt}
\caption{Experiment 3: Error versus time curve.}
\label{fig_exp3}
\vspace{-10pt}
\end{figure}
\subsection{Experiment 3}
In this experiment, we design a PR synthesis bank for a non-uniform analysis bank and test it with a non-stationary input signal. We have the following setup:
\begingroup
\allowdisplaybreaks
\setlength{\medmuskip}{0mu}
\setlength{\thickmuskip}{1mu}
\begin{align}
M_0 &= 2, & h_0(n) &= \{-0.1295, -0.12,0.3695,0.5018, \notag \\
& & & 0.3695,-0.12,-0.1295\} ,\notag \\
M_1 &= 3, & h_1(n) &= \{0.1308, 0.1728,-0.3775,0.2117, \notag \\
& & & 0.2117,-0.3775,0.1728, 0.1308\}, \notag \\
M_2 &= 6, & h_2(n) &= \{0.0717, 0.0749,-0.1148,0.1659, \notag \\ 
& & & -0.2069,0.2224,-0.2069, 0.1659,-0.1148, \notag \\ 
& & & 0.0749,0.0717\}, \notag \\
M_3 &= 6, & h_3(n) &= \{0.0881, 0.1617,-0.1686,-0.1538, \notag \\ 
& & & 0.1752,0.1752,-0.1538, -0.1686, \notag \\ 
& & & 0.1617,0.0881\}.
\end{align}
\endgroup
The analysis filters are designed using Parks-McClellan algorithm. We convert the NUFB into an equivalent UFB with the help of blocking operation. The resulting UFB has decimation factor equal to 6 and the following analysis filters:
\begingroup
\allowdisplaybreaks
\setlength{\thickmuskip}{0mu}
\setlength{\medmuskip}{0mu}
\begin{align}
g_0(n) &= h_0(n), & g_1(n) &= h_0(n-2), & g_2(n) &= h_0(n-4), \notag \\
g_3(n) &= h_1(n), & g_4(n) &= h_1(n-3), & g_5(n) &= h_2(n), \notag \\
g_6(n) &= h_3(n). & & & & 
\end{align}
\endgroup
We choose the length of the synthesis filter to be 7 and delay $d = 0$. The resulting matrix $\mathbf{K}$ satisfies the PR condition~(\ref{eq_k_pr_cond}) and using~(\ref{eq_genPR}), we obtain a PR solution with $\mathbf{w} = \mathbf{0}$. We apply a non-stationary signal generated using $\texttt{randn}(n) \sin(0.1n^2)$ to the filter bank and evaluate the reconstruction error as (refer Fig.~\ref{fig_wfsynth}):
\begin{equation}
e(n) = |\hat{u}(n) - d(n-M+1)|.
\end{equation}
Fig.~\ref{fig_exp3} shows the resulting error versus time curve. The small error which we obtain is due to finite precision in MATLAB.

\section{Conclusion}
\label{sec_conclusion}
We have addressed an important problem of designing an optimal FIR synthesis bank for a given non-maximally decimated analysis bank. To design such a synthesis bank, a framework was developed to apply matrix Wiener filtering to the problem and solution is obtained for a given length of synthesis filter and delay in the reconstructed output. We experimentally showed that better MSE can be obtained in some cases by adding appropriate delay in the reconstructed output. We found that the reconstruction error keeps decreasing as more and more multirate observations are added for a signal. The matrix representation developed in this paper can be utilized to carry out time-domain analysis of any multirate circuit. The infinite number of PR or non-PR synthesis bank solutions possible with our design can be optimized for a desired property in a future work.

% or


%\appendices
%% you can choose not to have a title for an appendix
%% if you want by leaving the argument blank
%\section{}

\appendix[FIR Matrix Wiener Filter Derivation]
\label{app_matder}
Consider an optimum matrix filtering problem as shown in Fig.~\ref{fig_opt} where the vector signal $\mathbf{v}(n)$ is the input and the signal $\mathbf{d}(n)$ is an estimate of the desired signal. In order to obtain the Wiener solution, the input and desired signals are assumed to be jointly stationary~\citep{bk:hayes}. The Wiener solution can be derived for either FIR or IIR case. Here we derive the FIR Wiener solution for a matrix filter whose constituent filters are causal and of length $P$. 
\begin{figure}[!t]
\centering
\includegraphics[scale = 0.7]{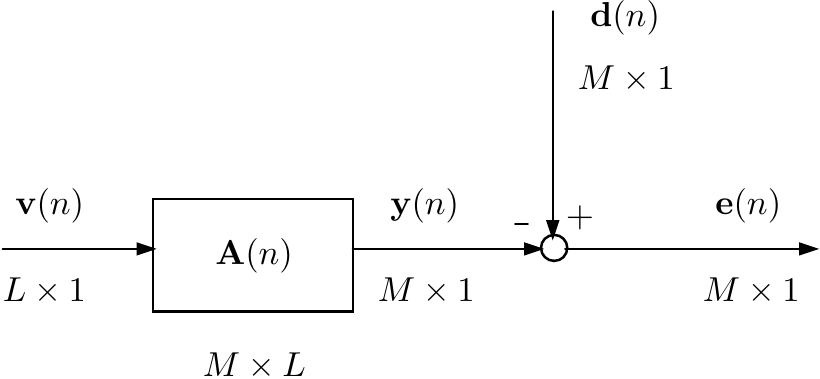}
\vspace{-8pt}
\caption{A general matrix Wiener filtering setup}
\label{fig_opt}
\vspace{-12pt}
\end{figure}
We begin by applying the principle of orthogonality~\cite{bk:haykin} to the setup and obtain
\begin{equation}
\allowdisplaybreaks
\medmuskip = 0mu
\begin{split}
\mathbf{E}[e_i(n) v_j^*(n-k)] &= 0, \qquad 0 \leq i < M; \\
& \quad 0 \leq j < L; \qquad   0 \leq k < P.
\end{split}
\end{equation}
The above equation can be written as
\begin{equation}
\thickmuskip = 3mu
\medmuskip=0mu
\begin{split}
\mathbf{E}[d_i(n) v_j^*(n-k)] &= \mathbf{E}[y_i(n) v_j^*(n-k)] \\
&= \sum_{r=0}^{L-1} \sum_{l=0}^{P-1} a_{i,r}(l) \mathbf{E}[v_r(n-l)v_j^*(n-k)].
\end{split}
\end{equation}
As the input and desired signals are jointly stationary, we can write
\begin{equation}
r_{d_i v_j}(k) = \sum_{r=0}^{L-1} \sum_{l=0}^{P-1} a_{i,r}(l) r_{v_r v_j}(k-l) .
\end{equation}
The above equation can be expressed in matrix form as
\begin{equation}
\thickmuskip = 2mu
\mathbf{r}_{d_i v_j} = \sum_{r=0}^{L-1} \mathbf{a}_{i,r}^T \mathbf{R}_{v_r v_j}, \qquad 0 \leq i < M; \qquad 0 \leq j < L, 
\end{equation}
where
\begingroup
\setlength{\thickmuskip}{0mu}
\setlength{\medmuskip}{0mu}
\setlength{\arraycolsep}{2pt}
\begin{align}
\mathbf{r}_{d_i v_j} &= \begin{bmatrix}
r_{d_i v_j}(0) & r_{d_i v_j}(1) & \dots & r_{d_i v_j}(P-1)
\end{bmatrix}, \notag \\
\mathbf{a}_{i,r} &= \begin{bmatrix}
a_{i,r}(0) & a_{i,r}(1) & \dots & a_{i,r}(P-1) 
\end{bmatrix}^T, \notag \\
\mathbf{R}_{v_r v_j} &= \begin{bmatrix}
r_{v_r v_j}(0) & r_{v_r v_j}(1) & \dots & r_{v_r v_j}(P-1) \\
r_{v_r v_j}(-1) & r_{v_r v_j}(0) & \dots & r_{v_r v_j}(P-2) \\
\vdots & \vdots & \ddots & \vdots \\
r_{v_r v_j}(-P+1) & r_{v_r v_j}(-P+2) & \dots & r_{v_r v_j}(0)
\end{bmatrix}.
\end{align}
\endgroup
A further compact form is
\begin{equation}
\begin{split}
\mathbf{r}_{d_i v} = \mathbf{a}_{i}^T \mathbf{R}_{v v}, \qquad 0 \leq i \leq M-1,
\end{split}
\end{equation}
where
\begin{align}
%\begin{split}
\mathbf{r}_{d_i v} &= \begin{bmatrix}
\mathbf{r}_{d_i v_0} & \mathbf{r}_{d_i v_1} & \dots & \mathbf{r}_{d_i v_{L-1}}
\end{bmatrix}, \notag \\
\mathbf{a}_{i} &= \begin{bmatrix}
\mathbf{a}_{i,0}^T & \mathbf{a}_{i,1}^T & \dots & \mathbf{a}_{i,L-1}^T 
\end{bmatrix}^T, \notag \\
\mathbf{R}_{v v} &= \begin{bmatrix}
\mathbf{R}_{v_0 v_0} & \mathbf{R}_{v_0 v_1} & \dots & \mathbf{R}_{v_0 v_{L-1}} \\
\mathbf{R}_{v_1 v_0} & \mathbf{R}_{v_1 v_1} & \dots & \mathbf{R}_{v_1 v_{L-1}} \\
\vdots & \vdots & \ddots & \vdots \\
\mathbf{R}_{v_{L-1} v_0} & \mathbf{R}_{v_{L-1} v_1} & \dots & \mathbf{R}_{v_{L-1} v_{L-1}}
\end{bmatrix} .
%\end{split}
\end{align}
The obtained result is a system of equations containing $LP$ equations in terms of $LP$ unknowns. These are the Wiener-Hopf equations for the matrix Wiener filter. Using them, we can determine the $i$-th row of the matrix Wiener filter.
%\section*{Acknowledgment}


\ifCLASSOPTIONcaptionsoff
  \newpage
\fi



\bibliographystyle{IEEEtran}
%\bibliography{ieee_timedomain}

%\begin{thebibliography}{1}
%\end{thebibliography}


%\end{IEEEbiography}
\vskip -1.5\baselineskip plus -1fil
\begin{IEEEbiography}[{\includegraphics[width=1in,height=1.25in,clip,keepaspectratio]{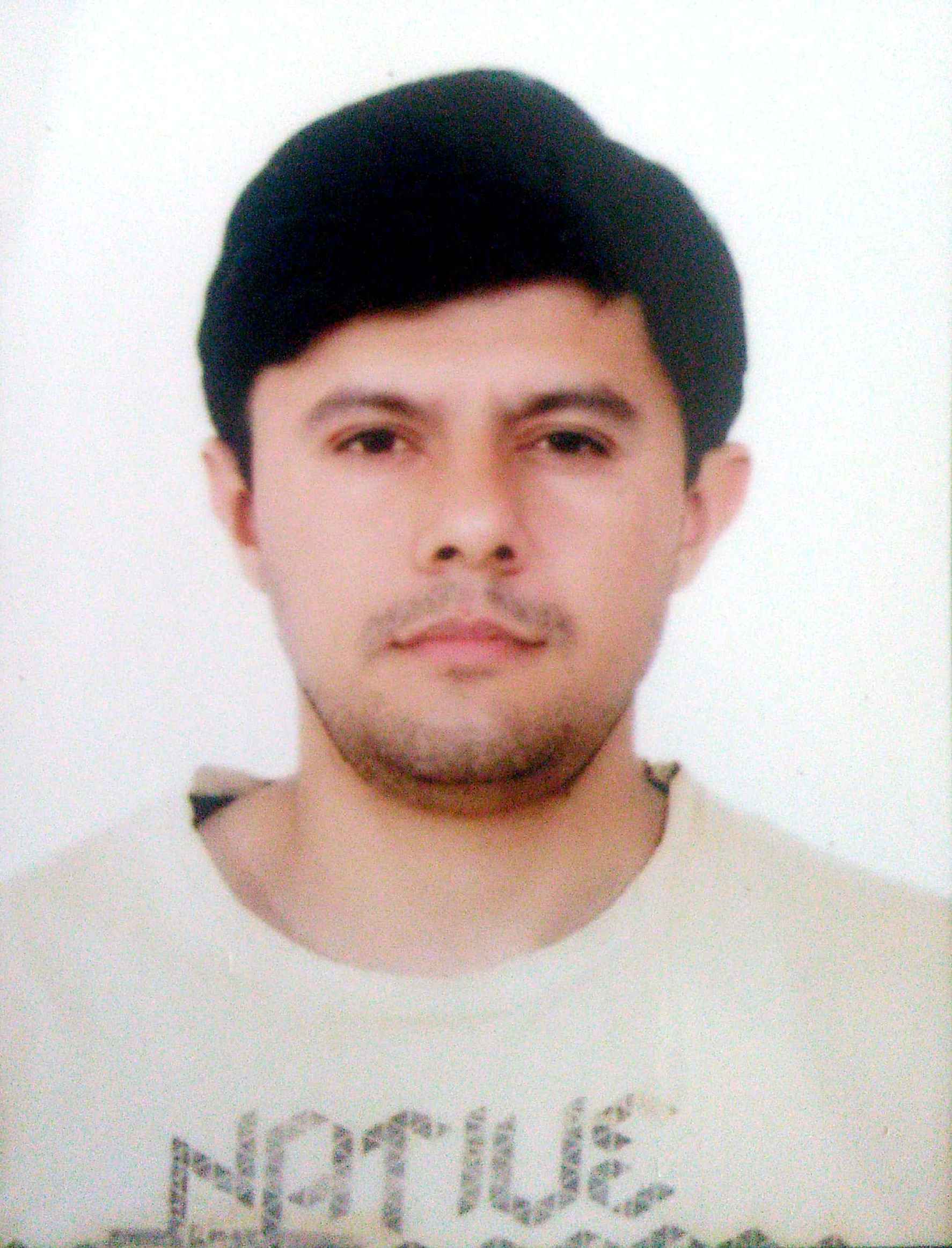}}]{Sandeep Patel} received the B.Tech degree in Electronics and Communications Engineering from the Indian Institute of Technology, Roorkee, in 2007. After completing graduation, he joined Nvidia Graphics and worked on video technologies till June, 2008. From June, 2008 to July, 2010, he was with Adobe Systems working in the flash media server team. He received the MS(R) degree from the Indian Insitute of Technology, Delhi in 2013. He is currently pursuing his Ph.D. degree from the Bharti School of Telecom Technology and Management, Indian Institute of Technology, Delhi.
\end{IEEEbiography}
\vskip -1.5\baselineskip plus -1fil
\begin{IEEEbiography}[{\includegraphics[width=1in,height=1.25in,clip,keepaspectratio]{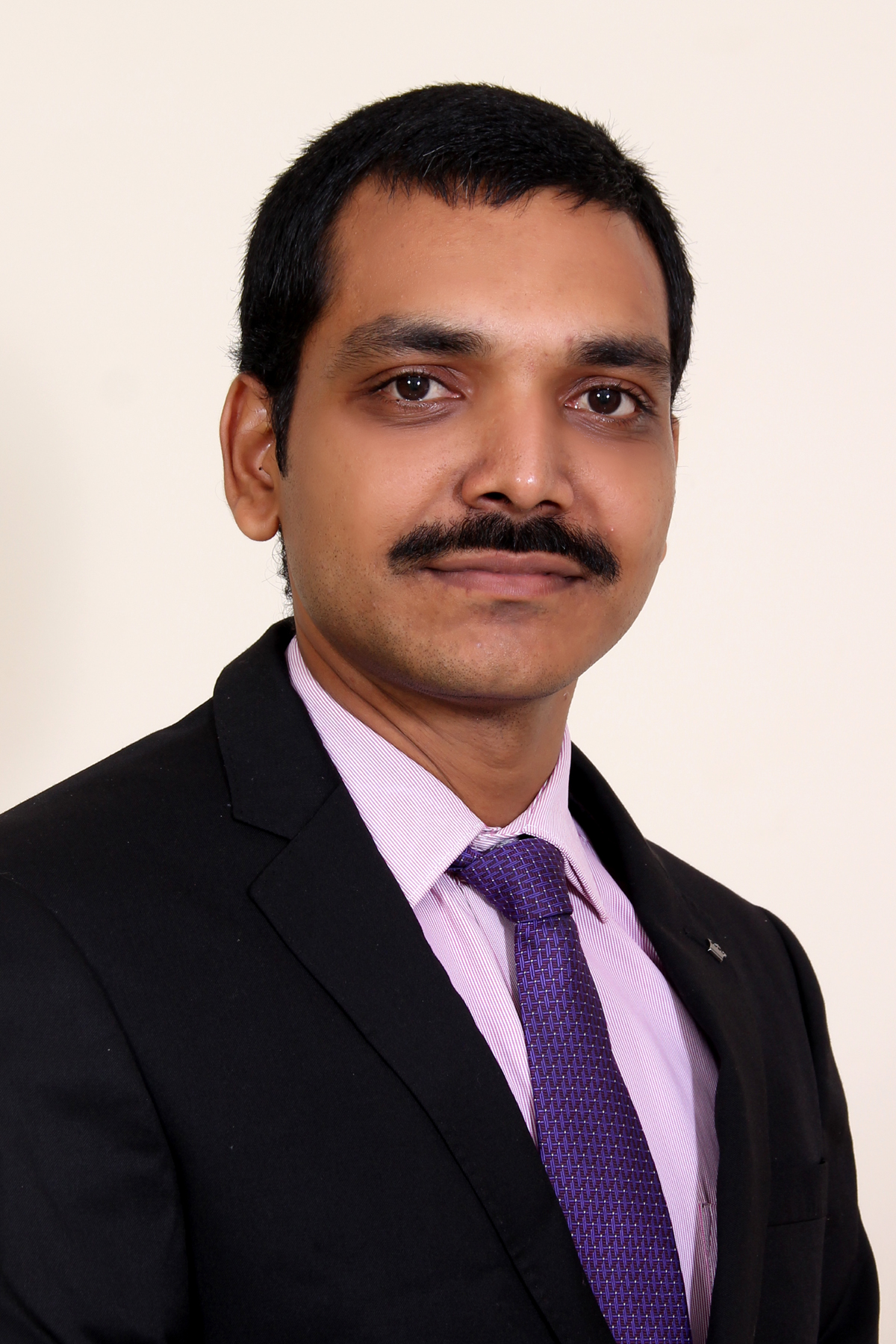}}]{Ravindra Dhuli} received the Ph.D. degree in signal processing from the Department of Electrical Engineering, Indian Institute of Technology Delhi, in 2010. He is currently an Associate Professor with the Department of Electronics and Communications Engineering, Vellore Institute of Technology, Andhra Pradesh, India. His research interests include multirate signal processing, statistical signal processing, image processing, and mathematical modeling.
\end{IEEEbiography}
\vskip -1.5\baselineskip plus -1fil
\begin{IEEEbiography}[{\includegraphics[width=1in,height=1.25in,clip,keepaspectratio]{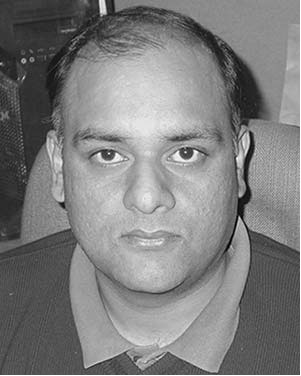}}]{Brejesh Lall} (M'05) received the B.E. and M.E. degrees in electronics and communications engineering from Delhi College of Engineering, Delhi University, in 1991 and 1992, respectively, and the Ph.D. degree in the area of multirate signal processing from the Indian Institute of Technology, Delhi, in 1999. His research interests are in signal processing with application to image processing and communication systems. He has over 150 papers in international journals and refereed conferences. From September 1997 to June 2005, he worked at Hughes Software Systems, in the digital signal processing group. Since July 2005, he has been with the faculty of the Department of Electrical Engineering, Indian Institute of Technology, Delhi, where he is currently a Professor and head of the Bharti School of Telecom Technology and Management. He is also the coordinator of Ericsson 5G Center of Excellence at IIT Delhi. 
\end{IEEEbiography}
%% insert where needed to balance the two columns on the last page with
%%\newpage


%\vfill


\end{document}